\documentclass[11pt]{article}

\usepackage{amsthm}
\usepackage{mathtools}
\usepackage{amsfonts}
\usepackage[utf8]{inputenc}

\usepackage{booktabs}
\usepackage[a4paper, total={8.5in, 11in},margin=1in]{geometry}

\usepackage[labelfont=bf]{caption}
\usepackage{subcaption}
\usepackage{comment}
\usepackage{dsfont}
\usepackage{amssymb}
\usepackage{bbm}
\usepackage{blkarray}
\usepackage[dvipsnames]{xcolor}
\usepackage[ruled,vlined]{algorithm2e}
\usepackage{multirow}
\usepackage{booktabs}
\usepackage{dcolumn}
\usepackage{units}
\usepackage{array}
\usepackage{tablefootnote}

\definecolor{custom_color}{RGB}{246,235,235}
\usepackage[framemethod=tikz]{mdframed}

\usepackage[colorlinks=true,linkcolor=Maroon,citecolor=Blue,urlcolor=black]{hyperref}
\usepackage{cleveref}
\DeclarePairedDelimiterX{\infdivx}[2]{(}{)}{%
  #1\;\delimsize\|\;#2%
}

\newtheorem{theorem}{Theorem}[section]
\newtheorem{definition}[theorem]{Definition}
\newtheorem{proposition}[theorem]{Proposition}

\newtheorem{conjecture}[theorem]{Conjecture}
\newtheorem{lemma}[theorem]{Lemma}
\newtheorem{observation}[theorem]{Observation}
\newtheorem{claim}[theorem]{Claim}

\newenvironment{sproof}{%
  \proof}{\endproof}

\SetAlgorithmName{Mechanism}{mechanism}{List of Mechanisms}

\DeclareMathOperator*{\argmax}{arg\,max}

\DeclareMathOperator{\distortion}{distortion}
\DeclareMathOperator{\define}{\overset{def}{=}}
\DeclareMathOperator{\poly}{poly}

\DeclareMathOperator{\plu}{plu}
\DeclareMathOperator{\topp}{top}
\DeclareMathOperator{\secc}{sec}
\DeclareMathOperator{\veto}{veto}

\DeclareMathOperator{\tv}{d_{TV}}
\DeclareMathOperator{\SC}{\textsc{SC}}
\DeclareMathOperator{\DR}{\textsc{DominationRoot}}

\DeclareMathOperator{\minimax}{\textsc{Minimax-LP}}
\DeclareMathOperator{\minimaxalpha}{\textsc{Minimax}^{\alpha}\textsc{-LP}}
\DeclareMathOperator{\pluralitymatching}{\textsc{PluralityMatching}}
\DeclareMathOperator{\metricalphaLP}{\textsc{Metric}^{\alpha}{-LP}}
\DeclareMathOperator{\metricLP}{\textsc{Metric-LP}}

\author{
  Ioannis Anagnostides\\[-2mm]
  National Technical University of Athens\\[-2mm]
  \texttt{ioannis.anagnostides@gmail.com}
  \and
  Dimitris Fotakis\\[-2mm]
  National Technical University of Athens\\[-2mm]
  \texttt{fotakis@cs.ntua.gr}
  \and 
  Panagiotis Patsilinakos\\[-2mm]
  National Technical University of Athens\\[-2mm]
  \texttt{patsilinak@corelab.ntua.gr}
}

\date{\today}                     

\title{Metric-Distortion Bounds under Limited Information}

\begin{document}

\maketitle
\pagenumbering{gobble}

\begin{abstract}
    In this work we study the metric distortion problem in voting theory under a limited amount of ordinal information. Our primary contribution is threefold. First, we consider mechanisms which perform a sequence of pairwise comparisons between candidates. We show that a widely-popular deterministic mechanism employed in most knockout phases yields distortion $\mathcal{O}(\log m)$ while eliciting only $m-1$ out of $\Theta(m^2)$ possible pairwise comparisons, where $m$ represents the number of candidates. Our analysis for this mechanism leverages a powerful technical lemma recently developed by Kempe \cite{DBLP:conf/aaai/000120a}. We also provide a matching lower bound on its distortion. In contrast, we prove that any mechanism which performs fewer than $m-1$ pairwise comparisons is destined to have unbounded distortion. Moreover, we study the power of deterministic mechanisms under incomplete rankings. Most notably, when every agent provides her $k$-top preferences we show an upper bound of $6 m/k + 1$ on the distortion, for any $k \in \{1, 2, \dots, m\}$. Thus, we substantially improve over the previous bound of $12 m/k$ recently established by Kempe \cite{DBLP:conf/aaai/000120a,DBLP:conf/aaai/000120b}, and we come closer to matching the best-known lower bound. Finally, we are concerned with the sample complexity required to ensure near-optimal distortion with high probability. Our main contribution is to show that a random sample of $\Theta(m/\epsilon^2)$ voters suffices to guarantee distortion $3 + \epsilon$ with high probability, for any sufficiently small $\epsilon > 0$. This result is based on analyzing the sensitivity of the deterministic mechanism introduced by Gkatzelis, Halpern, and Shah \cite{DBLP:conf/focs/Gkatzelis0020}. Importantly, all of our sample-complexity bounds are distribution-independent. 
    
    From an experimental standpoint we present several empirical findings on real-life voting applications, comparing the scoring systems employed in practice with a mechanism explicitly minimizing (metric) distortion. Interestingly, for our case studies we find that the winner in the actual competition is typically the candidate who minimizes the distortion. 
\end{abstract}

\clearpage
\pagenumbering{arabic}

\section{Introduction}

Aggregating the preferences of individuals into a collective decision lies at the heart of social choice, and has recently found numerous applications in areas such as information retrieval, recommender systems, and machine learning \cite{DBLP:journals/jmlr/VolkovsZ14,DBLP:conf/cikm/VolkovsLZ12,Beliakov2011,10.5555/1941934,10.1007/11510888_21}. According to the classic theory of Von Neumann and Morgenstern \cite{vonNeumann1944-VONTOG-4} individual preferences are captured through a \emph{utility function}, which assigns numerical (or \emph{cardinal}) values to each alternative. Yet, in voting theory, as well as in most practical applications, mechanisms typically elicit only \emph{ordinal} information from the voters, indicating an order of preferences over the candidates. Although this might seem at odds with a utilitarian framework, it has been recognized that it might be hard for a voter to specify a precise numerical value for an alternative, and providing only ordinal information substantially reduces the cognitive burden. This begs the question: What is the loss in efficiency of a mechanism extracting only ordinal information with respect to the \emph{utilitarian social welfare}, i.e. the sum of individual utilities over a chosen candidate? The framework of \emph{distortion} introduced by Procaccia and Rosenschein \cite{DBLP:conf/cia/ProcacciaR06} measures exactly this loss from an approximation-algorithms standpoint, and has received considerable attention in recent years. 

As it turns out, the approximation-guarantees we can hope for crucially depend on the assumptions we make on the utility functions. For example, in the absence of any structure Procaccia and Rosenschein observed that \emph{no} ordinal mechanism can obtain bounded distortion \cite{DBLP:conf/cia/ProcacciaR06}. In this work we focus on the \emph{metric distortion} framework, introduced by Anshelevich et al. \cite{DBLP:conf/aaai/AnshelevichBP15}, wherein voters and candidates are thought of as points in some arbitrary metric space; this is akin to some models in spatial voting theory \cite{10.2307/30025956}. In this context, the voters' preferences are measured by means of their ``proximity'' from each candidate, and the goal is to output a candidate who (approximately) minimizes the \emph{social cost}, i.e. the cumulative distances from the voters. A rather simplistic view of this framework manifests itself when agents and candidates are embedded into a one-dimensional line, and their locations indicate whether they are ``left'' or ``right'' on the political spectrum. However, the metric distortion framework has a far greater reach since no assumptions whatsoever are made for the underlying metric space. Indeed, the dimensionality of the space is potentially unbounded, while we are not even restricted in Euclidean spaces.

Importantly, this paradigm offers a compelling way to quantitatively compare different voting rules commonly employed in practice \cite{DBLP:conf/aaai/SkowronE17,DBLP:conf/aaai/000120a,10.1145/3033274.3085138,DBLP:conf/aaai/AnshelevichBP15}, while it also serves as a benchmark for designing new mechanisms in search of better distortion bounds \cite{DBLP:conf/focs/Gkatzelis0020,DBLP:conf/ec/MunagalaW19}. A common assumption made in this line of work is that the algorithm has access to the entire total rankings of the voters. However, there are many practical scenarios in which it might be desirable to truncate the ordinal information elicited by the mechanism. For example, requesting only the top preferences could further relieve the cognitive burden since it might be hard for a voter to compare alternatives which lie on the bottom of her list (for additional motivation for considering incomplete or partial orderings see \cite{DBLP:conf/aistats/FotakisKS21,CHEN2013521,BENFERHAT200425}, and references therein), while any truncation in the elicited information would also translate to more efficient communication. These reasons have driven several authors to study the decay of distortion under missing information \cite{DBLP:conf/aaai/000120b,DBLP:conf/ijcai/AnshelevichP16,DBLP:conf/aaai/GrossAX17,DBLP:conf/aaai/FainGMP19}, potentially allowing some randomization (see our related work subsection). In this work we follow this line of research, offering several new insights and improved bounds over prior results.

\subsection{Contributions \& Techniques}

First, we study voting rules which perform a sequence of pairwise comparisons between two candidates, with the result of each comparison being determined by the majority rule over the entire population of voters. This class includes many common mechanisms such as Copeland's rule \cite{10.2307/25054952} or the minimax scheme of Simpson and Kramer \cite{10.1257/jep.9.1.3}, and has received considerable attention in the literature of social choice; cf., see \cite{DBLP:conf/ijcai/LangPRVW07}, and references therein. Within the framework of (metric) distortion the following fundamental question arises:

\begin{quote}
    \textit{How many pairwise comparisons between two candidates are needed to guarantee non-trivial bounds on the distortion?}
\end{quote}

For example, Copeland's rule (and most of the common voting rules within this class) elicits all possible pairwise comparisons, i.e. $\binom{m}{2} = \Theta(m^2)$, and guarantees distortion at most $5$ \cite{DBLP:conf/aaai/AnshelevichBP15}. Thus, it is natural to ask whether we can substantially truncate the number of elicited pairwise comparisons without sacrificing too much the efficiency of the mechanism. We stress that we allow the queries to be dynamically adapted during the execution of the algorithm. In this context, we provide the following strong positive result:

\begin{theorem}
    There exists a deterministic mechanism which elicits only $m-1$ pairwise comparisons and guarantees distortion $\mathcal{O}(\log m)$.
\end{theorem}

Our mechanism is particularly simple and natural: In every round we arbitrarily pair the remaining candidates and we only extract the corresponding comparisons. Next, we eliminate all the candidates who lost and we continue recursively until a single candidate emerges victorious. Interestingly, this mechanism is widely employed in practical applications, for example in the knockout phases of most competitions, with the difference that typically some ``prior'' is used in order to construct the pairings. The main technical ingredient of the analysis is a powerful lemma developed by Kempe via an LP duality argument \cite{DBLP:conf/aaai/000120a}. Specifically, Kempe characterized the social cost ratio between two candidates when there exists a sequence of intermediate alternatives such that every candidate in the chain pairwise-defeats the next one. We also supplement our analysis for this mechanism with a matching lower bound on a carefully constructed instance. Moreover, we show that any mechanism which performs (strictly) fewer than $m-1$ pairwise comparisons has \emph{unbounded} distortion. This limitation applies even if we allow randomization either during the elicitation or the winner-determination phase. Indeed, there are instances for which only a single alternative can yield bounded distortion, but the mechanism simply does not have enough information to identify the ``right'' candidate.

Next, we study deterministic mechanisms which only receive an incomplete order of preferences from every voter, instead of the entire rankings as it is usually assumed. This setting has already received attention in the literature, most notably by Kempe \cite{DBLP:conf/aaai/000120b}, and has numerous applications in real-life voting systems. Arguably the most important such consideration arises when every voter provides her $k$-top preferences, for some parameter $k \in [m]$. Kempe showed \cite{DBLP:conf/aaai/000120b} that there exists a deterministic mechanism which elicits only the $k$-top preferences and whose distortion is upper-bounded by $79m/k$; using a powerful tool developed in \cite{DBLP:conf/aaai/000120a} this bound can be improved all the way down to $12m/k$. However, this still leaves a substantial gap with respect to the best-known lower bound, which is $2m/k$ if we ignore some additive constant factors. Thus, Kempe \cite{DBLP:conf/aaai/000120b} left as an open question whether the aforementioned upper bound can be improved. In our work we make substantial progress towards bridging this gap, proving the following:

\begin{theorem}
    There exists a deterministic mechanism which only elicits the $k$-top preferences and yields distortion at most $6m/k + 1$.
\end{theorem}

We should stress that the constant factors are of particular importance in the framework of metric distortion; indeed, closing the gap even for the special case of $k = m$ has received intense scrutiny in recent years \cite{DBLP:conf/aaai/AnshelevichBP15,DBLP:conf/ec/MunagalaW19,DBLP:conf/aaai/000120a,DBLP:conf/focs/Gkatzelis0020}. From a technical standpoint the main technique for proving such upper bounds consists of identifying a candidate for which there exists a path to any other node such that every candidate in the path pairwise-defeats the next one by a sufficiently large margin (which depends on $k$). Importantly, the derived upper bound crucially depends on the length of the path. Our main technical contribution is to show that there always exists a path of length $2$ with the aforedescribed property, while the previous best result by Kempe established the claim only for paths of length $3$.

Although our approach can potentially bring further improvements, closing the gap inevitably requires different techniques. In particular, a promising direction appears to stem from extending some of the claims established by Gkatzelis et al. \cite{DBLP:conf/focs/Gkatzelis0020}. Indeed, we observe that a natural generalization of the main technical ingredient in \cite{DBLP:conf/focs/Gkatzelis0020} would lower the upper bound to $4 m/k - 1$ (\Cref{proposition:conditional}), which appears to be optimal when $k$ is close to $m$. More precisely, the authors in \cite{DBLP:conf/focs/Gkatzelis0020} proved that a certain graph always has a perfect matching when the entire rankings are available; we conjecture that under $k$-top preferences there always exists a perfect matching for a subset of a $k/m$ fraction of the voters (see \Cref{conjecture:extension} for a more precise statement). 

We also provide some other interesting bounds for deterministic mechanisms under missing information. Most notably, if the voting rule performs well on an arbitrary (potentially adversarially selected) subset of the voters can we quantify its distortion over the entire population? We answer this question with a sharp upper bound in \Cref{theorem:missing_voters}. In fact, we use this result as a tool for some of our other proofs, but nonetheless we consider it to be of independent interest. It should be noted that even in the realm of partial or incomplete rankings there exists an \emph{instance-optimal} mechanism via linear programming; this was first observed by Goel et al. \cite{10.1145/3033274.3085138} when the total orders are available, but it directly extends in more general settings. Interestingly, we show that the recently introduced mechanism of Gkatzelis et al. \cite{DBLP:conf/focs/Gkatzelis0020} which always obtains distortion at most $3$ can be substantially outperformed by the LP mechanism. Namely, for some instances the mechanism of Gkatzelis et al. \cite{DBLP:conf/focs/Gkatzelis0020} yields distortion almost $3$ while the instance-optimal mechanism yields distortion close to $1$. 

Finally, we consider mechanisms which receive information from only a ``small'' \emph{random sample} of voters; that is, we are concerned with the \emph{sample complexity} required to ensure efficiency, which boils down to the following fundamental question:

\begin{quote}
    \textit{How large should the size of the sample be in order to guarantee near-optimal distortion with high probability?}
\end{quote}

More precisely, we are interested in deriving sample-complexity bounds which are \emph{independent} on the number of voters $n$. This endeavor is particularly motivated given that in most applications $n \gg m$. Naturally, sampling approximations are particularly standard in the literature of social choice. Indeed, in many scenarios one wishes to predict the outcome of an election based on small sample (e.g. in polls or exit polls), while in many other applications it is considered even infeasible to elicit the entire input (e.g. in online surveys). In this context, we will be content with obtaining near-optimal distortion \emph{with high probability} (e.g. $99\%$). This immediately deviates from the line of research studying randomized mechanisms (cf. see \cite{DBLP:conf/ijcai/AnshelevichP16}) wherein it suffices to obtain a guarantee in expectation. We point out that it has been well-understood that a guarantee only in expectation might be insufficient in many cases; for example, Fain et al. \cite{DBLP:conf/aaai/FainGMP19} considered as the objective the \emph{squared distortion} as a proxy in order to limit as much as possible the variance in the distortion. In fact, the authors in \cite{DBLP:conf/aaai/FainGMP19} are also concerned with sample complexity issues, but from a very different standpoint.

We stress that we only allow randomization during the preference elicitation phase; for a given random sample, which corresponds to the \emph{entire} rankings of the voters, the mechanisms we consider act deterministically. Specifically, we analyze two main voting rules along this vein.

\begin{theorem}[Approximate Copeland]
    For any sufficiently small $\epsilon > 0$ there exists a mechanism which takes a sample of size $\Theta(\log(m)/\epsilon^2)$ voters and yields distortion at most $5 + \epsilon$ with probability $0.99$.
\end{theorem}

The techniques required for the proof of this theorem are fairly standard. More importantly, we analyze the sample complexity of $\pluralitymatching$, the mechanism of Gkatzelis et al. \cite{DBLP:conf/focs/Gkatzelis0020} which recovers the optimal distortion bound of $3$ (among deterministic mechanisms). In this context, we establish the following result:

\begin{theorem}[Approximate $\pluralitymatching$]
    For any sufficiently small $\epsilon > 0$ there exists a mechanism which takes a sample of size $\Theta(m/\epsilon^2)$ voters and yields distortion at most $3 + \epsilon$ with probability $0.99$.
\end{theorem}

More precisely, the main ingredient of $\pluralitymatching$ is a maximum-matching subroutine for a certain bipartite graph. Our first observation is that the size of the maximum matching can be determined through a much smaller graph which satisfies a ``proportionality'' condition with respect to a maximum-matching decomposition. Although this condition cannot be explicitly met since the algorithm is agnostic to the decomposition, our observation is that sampling (with sufficiently many samples) will approximately satisfy this requirement, eventually leading to the desired conclusion. 

We stress that we do \emph{not} guarantee that the winner in our sample will coincide with that over the entire population. In fact, the sample complexity bounds for the \emph{winner determination} problem---for virtually every reasonable voting rule---depend on the \emph{margin} of victory (see \cite{10.5555/2772879.2773334}); however, we argue that this feature is undesirable. For one thing the algorithm does not have any prior information on the margin, and hence it is unclear how to tune this parameter in practice. More importantly, in many scenarios the margin might be very small, leading to a substantial overhead in the sample-complexity requirements of the mechanism. One of our conceptual contributions is to show that we can circumvent such limitations once we espouse a utilitarian framework. Indeed, observe that all of our bounds are \emph{distribution-independent} (and instance-oblivious).

We should also point out that although we are emphasizing on sample-complexity considerations, we believe that our results have another very clear motivation. Namely, given that in most applications $n \gg m$, it is important to provide \emph{sublinear} algorithms whose running time does not depend on $n$. In this context, we provide a Monte Carlo implementation of $\pluralitymatching$ whose time complexity scales independently of $n$.

To conclude, we provide several experimental findings in real-life voting applications from the standpoint of the (metric) distortion framework. We are mostly concerned with comparing the results of the scoring systems used in practice against a mechanism which explicitly attempts to minimize the distortion; the latter is realized with the linear programming mechanism of Goel et al. \cite{10.1145/3033274.3085138}. Specifically, we analyze the efficiency of the scoring rule used in the Eurovision song contest and the Formula One world championship. Interestingly, on both occasions we find that the winner in the actual competition is usually the candidate who minimizes the distortion. Our implementation is publicly available at \texttt{\href{https://github.com/ioannisAnagno/Voting-MetricDistortion}{https://github.com/ioannisAnagno/Voting-MetricDistortion}}.

\subsection{Related Work}

Research in the metric distortion framework was initiated by Anshelevich et al. \cite{DBLP:conf/aaai/AnshelevichBP15}. Specifically, they analyzed the distortion of several common voting rules, most notably establishing that Copeland's rule has distortion at most $5$, with the bound being tight for certain instances. They also conjectured that the \emph{ranked pairs} mechanism always achieves distortion at most $3$, which is also the lower bound for any deterministic mechanism. This conjecture was disproved by Goel et al. \cite{10.1145/3033274.3085138},\footnote{A tight bound of $\Theta(\sqrt{m})$ for the ranked pairs mechanism was subsequently given by Kempe \cite{DBLP:conf/aaai/000120a}.} while they also studied \emph{fairness} properties of certain voting rules. Moreover, Skowron and Elkind \cite{DBLP:conf/aaai/SkowronE17} established that a popular rule named \emph{single transferable vote} (STV) has distortion $\mathcal{O}(\log m)$, along with a nearly-matching lower bound. The barrier of $5$ set out by Copeland was broken by Munagala and Wang \cite{DBLP:conf/ec/MunagalaW19}, presenting a novel deterministic rule with distortion $2 + \sqrt{5}$. The same bound was obtained by Kempe \cite{DBLP:conf/aaai/000120a} through an LP duality framework, who also articulated sufficient conditions for proving the existence of a deterministic mechanism with distortion $3$. This conjecture was only recently confirmed by Gkatzelis et al. \cite{DBLP:conf/focs/Gkatzelis0020}, introducing the \emph{plurality matching} mechanism. Closely related to our study is also the work of Gross et al. \cite{DBLP:conf/aaai/GrossAX17}, wherein the authors provide a near-optimal mechanism which only asks $m+1$ voters for their top-ranked alternatives. The main difference with our setting is that we require an efficiency-guarantee with high probability, and not in expectation.

\paragraph{Broader Context.} Beyond the metric case most focus has been on analyzing distortion under a \emph{unit-sum} assumption on the utility function, ensuring that agents have equal ``weights''. In particular, Boutilier et al. \cite{DBLP:journals/ai/BoutilierCHLPS15} provide several upper and lower bounds, while they also study learning-theoretic aspects under the premise that every agent's utility is drawn from a distribution (cf., see \cite{DBLP:journals/ai/ProcacciaZPR09}). Moreover, several multi-winner extensions have been studied in the literature. Caragiannis et al. \cite{DBLP:conf/ijcai/CaragiannisNP016} studied the \emph{committee selection} problem, which consists of selecting $k$ alternatives that maximize the social welfare, assuming that the value of each agent is defined as the maximum value derived from the committee's members. We also refer to \cite{DBLP:conf/aaai/BenadeNP017} for the \emph{participatory budgeting} problem, and to \cite{DBLP:conf/aaai/BenadePQ19} when the output of the mechanism should be a total order over alternatives (instead of a single winner).

More special cases were considered in \cite{DBLP:conf/sigecom/FeldmanFG16,DBLP:conf/wine/FainGMS17}, strengthening some of the results we previously described. The trade-off between efficiency and communication has been addressed in \cite{DBLP:conf/nips/MandalP0W19,10.1145/3391403.3399510}, while Amanatidis et al. \cite{DBLP:conf/aaai/AmanatidisBFV20} investigated the decay of distortion under a limited amount of cardinal queries---in addition to the ordinal information. We should also note a series of works analyzing the power of ordinal preferences for some fundamental graph-theoretic problems \cite{DBLP:conf/sagt/Filos-RatsikasF014,DBLP:conf/aaai/AnshelevichS16,DBLP:conf/wine/AnshelevichS16,DBLP:conf/wine/AnshelevichZ18}. Finally, we point out that strategic issues are typically ignored within this line of work. We will also posit that agents provide truthfully their preferences, but we refer to \cite{DBLP:conf/aaai/BhaskarDG18,DBLP:conf/wine/CaragiannisFFHT16,DBLP:journals/corr/abs-1802-01308} for rigorous considerations on the strategic issues that arise. We refer the interested reader to the survey of Anshelevich et al. \cite{anshelevich2021distortion}, as we have certainly not exhausted the literature.

\section{Preliminaries}

A \emph{metric space} is a pair $(\mathcal{M}, d)$, where $d: \mathcal{M} \times \mathcal{M} \mapsto \mathbb{R}$ constitutes a \emph{metric} on $\mathcal{M}$, i.e., (i) $\forall x,y \in \mathcal{M}, d(x, y) = 0 \iff x = y$ (identity of indiscernibles), (ii) $\forall x, y \in d(x,y) = d(y, x)$ (symmetry), and (iii) $\forall x, y, z \in \mathcal{M}, d(x, y) \leq d(x, z) + d(z, y)$ (triangle inequality). Consider a set of $n$ voters $V = \{1, 2, \dots, n\}$ and a set of $m$ candidates $C = \{a, b, \dots, \}$; candidates will be typically represented with lowercase letters such as $a, b, w, x$, but it will be sometimes convenient to use numerical values as well. We assume that every voter $i \in V$ is associated with a point $v_i \in \mathcal{M}$, and every candidate $a \in C$ to a point $c_a \in \mathcal{M}$. Our goal is to select a candidate $x$ who minimizes the \emph{social cost}: $\SC(x) = \sum_{i=1}^{n} d(v_i, c_{x})$. This task would be trivial if we had access to the agents' distances from all the candidates. However, in the standard \emph{metric distortion} framework every agent $i$ provides only a \emph{ranking} (a total order) $\sigma_i$ over the points in $C$ according to the \emph{order} of $i$'s distances from the candidates. We assume that ties are broken arbitrarily, subject to transitivity, but we will not abuse the tie-breaking assumption.

In this work we are considering a substantially more general setting, wherein every agent provides a subset of $\sigma_i$. More precisely, we assume that agent $i$ provides as input a set $\mathcal{P}_i$ of ordered pairs of distinct candidates, such that $(a, b) \in \mathcal{P}_i \implies a \succ_i b$, where $a, b \in C$; it will always be assumed that $\mathcal{P}_i$ corresponds to the \emph{transitive closure} of the input. We will allow $\mathcal{P}_i$ to be the empty set, in which case $i$ does not provide any information to the mechanism; with a slight abuse of notation we will let $\mathcal{P}_i \equiv \sigma_i$ when $i$ provides the entire order of preferences. We will say that the input $\mathcal{P} = (\mathcal{P}_1, \dots, \mathcal{P}_n)$ is consistent with the metric $d$ if $(a, b) \in \mathcal{P}_i \implies d(v_i, c_a) \leq d(v_i, c_b), \forall i \in V$, and this will be denoted with $d \triangleright \mathcal{P}$. We will represent with $\topp(i)$ and $\secc(i)$ $i$'s first and second most preferred candidates respectively. We may also sometimes use the notation $ab = \{ i \in V : a \succ_i b \}$.

A \emph{deterministic} \emph{social choice rule} is a function that  maps an \emph{election} in the form of a $3$-tuple $\mathcal{E} = (V, C, \mathcal{P})$ to a single candidate $a \in C$. We will measure the performance of $f$ for a given input of preferences $\mathcal{P}$ in terms of its \emph{distortion}, namely, the worst-case approximation ratio it provides with respect to the social cost:

\begin{equation}
    \distortion(f; \mathcal{P}) = \sup \frac{\SC(f(\mathcal{P}))}{\min_{a \in C} \SC(a)},
\end{equation}
where the supremum is taken over all metrics such that $d \triangleright \mathcal{P}$. The distortion of a social choice rule $f$ is the maximum of $\distortion(f; \mathcal{P})$ over all possible input preferences $\mathcal{P}$. In other words, once the mechanism selects a candidate (or a distribution over candidates if the social choice rule is \emph{randomized}) an adversary can select any metric space subject to being consistent with the input preferences. Similarly, in \Cref{section:pairwise_comparisons} where we study mechanisms which perform pairwise comparisons, the adversary can select any metric space consistent with the elicited comparisons. We should point out the following:

\begin{proposition}
    \label{proposition:feasibility}
Under any given preferences $\mathcal{P}$, there exists a metric space consistent with $\mathcal{P}$.
\end{proposition}

This proposition follows directly from Proposition 1 in \cite{DBLP:conf/aaai/AnshelevichBP15}, which established the claim when $\mathcal{P} = \sigma$. 

\subsection{Instance-Optimal Voting}

An important observation is that under any input preferences $\mathcal{P}$ there exists a deterministic \emph{instance-optimal} mechanism; this was noted by Goel et al. \cite{10.1145/3033274.3085138} (see also \cite{DBLP:journals/ai/BoutilierCHLPS15}) when $\mathcal{P} = \sigma$, but their mechanism directly applies for our more general setting. We briefly present their idea, as we will also employ this mechanism for our experiments.

The first ingredient is an optimization problem that allows to compare a pair of distinct candidates, subject to the set of preferences given to the mechanism. Specifically, for $a, b \in C$, with $a \neq b$, consider the following \emph{linear program} $\metricLP(a, b)$:

\begin{equation}
    \label{eq:linear_program}
\begin{array}{ll@{}ll}
\text{maximize}  & \sum_{i=1}^n x_{i, a} &\\
\text{subject to} & \sum_{i=1}^n x_{i, b} = 1; \\ 

                & x_{i, p} \leq x_{i, q}, & \forall (p, q) \in \mathcal{P}_i, \forall i \in V; \\
                 &x_{i, i} = 0, & \forall i \in V \cup C; \\
                 & x_{i, j} = x_{j, i}, & \forall i, j \in V \cup C; \\
                 & x_{i, j} \leq x_{i, k} + x_{k, j}, & \forall i, j, k \in V \cup C.
\end{array}
\end{equation}

First of all, observe that the output of this linear program does not necessarily yield a metric space, but rather a \emph{pseudo-metric} on $V \cup C$, given that it is possible that $x_{i, j} = 0$ for $i \neq j$. Nonetheless, this issue can be easily resolved by ``merging'' all the elements $i, j \in V \cup C$ such that $x_{i, j} = 0$; the consistency of this approach follows by the ``triangle inequalities''---the final constraint of the linear program. It should be pointed out that some of the constraints in $\metricLP$ are redundant, in the sense that they are implied by others, but we will not dwell on such optimizations here.

We will represent with $\mathfrak{D}(a | b)$ the value of the linear program $\metricLP(a,b)$; if it is unbounded we let $\mathfrak{D}(a | b) = +\infty$. Note that the linear program is always feasible by virtue of \Cref{proposition:feasibility}. In this context, the mechanism of Goel et al. \cite{10.1145/3033274.3085138} consists of the following steps:

\begin{itemize}
    \item For any pair $a, b \in C$, with $a \neq b$, compute $\mathfrak{D}(a|b)$; also let $\mathfrak{D}(a | a) = 1$. 
    \item Set $\mathfrak{D}(a) = \max_{b \in C} \mathfrak{D}(a|b)$.
    \item Return the candidate $b$ with the minimum value $\mathfrak{D}(a)$ over all $a \in C$; ties are broken arbitrarily.
\end{itemize}

This mechanism will be referred to as $\minimax$, to distinguish from the \emph{minimax} voting scheme of Simpson and Kramer \cite{10.1257/jep.9.1.3}. The $\minimax$ rule essentially performs brute-force search over all possible metrics in order to identify the candidate which minimizes the distortion; nonetheless, it can be solved in $\poly(n, m)$ time given that the $\metricLP$ admits a strongly polynomial time algorithm; this follows because the \emph{bit complexity} $L$---the number of bits required to represent it \cite{10.1145/800057.808695}---is small: $L = \mathcal{O}(\log (n+m))$. Moreover, it is easy to establish the following:

\begin{theorem}
    \label{theorem:instance_opt}
    For any given preferences $\mathcal{P}$, the $\minimax$ rule is instance-optimal in terms of distortion.
\end{theorem}

In particular, when $\mathcal{P} = \sigma$ note that $\minimax$ always yields distortion $3$ by virtue of an upper-bound by Gkatzelis et al. \cite{DBLP:conf/focs/Gkatzelis0020}. However, in this work we will be mostly interested in providing upper bounds on the distortion of $\minimax$ under incomplete rankings.

\section{Sequence of Pairwise Comparisons}
\label{section:pairwise_comparisons}

In this section we are considering voting rules which perform a sequence of pairwise comparisons between two candidates, with the result of each comparison being determined by the majority rule over the entire population of voters. To put it differently, consider the tournament graph $T = (C, E)$ where $(a,b) \in E$ if and only if candidate $a$ pairwise-defeats candidate $b$; it will be tacitly assumed---without any loss of generality---that ties are broken arbitrarily so that $T$ is indeed a tournament. We are studying mechanisms which elicit edges from $T$, and we are interested in establishing a trade-off between the number of elicited edges and the distortion of the mechanism. We commence with the following lower bound:

\begin{proposition}
    \label{proposition:disconnected}
    There are instances for which any deterministic mechanism which elicits (strictly) fewer than $m-1$ edges from $T$ has unbounded distortion.
\end{proposition}

\begin{sproof}
Consider a family of tournaments $\mathcal{T}$ as illustrated in \Cref{fig:disconnected}, with the set $C^*$ containing a single candidate. Then, there are metric spaces for which all the voters are arbitrarily close to the candidate in $C^*$ and arbitrarily far from any other candidate. Thus, every mechanism with bounded distortion has to identify the candidate in $C^*$. However, it is easy to see that any pairwise comparison can eliminate at most one candidate from being in $C^*$. As a result, if $\widehat{T} = (C, \widehat{E})$ is the subgraph based on the elicited edges, there will be at least two distinct candidates which could lie in $C^*$ for some tournament in $\mathcal{T}$ consistent with $\widehat{T}$, leading to the desired conclusion.
\end{sproof}

\begin{figure}[!ht]
    \centering
    \includegraphics[scale=0.35]{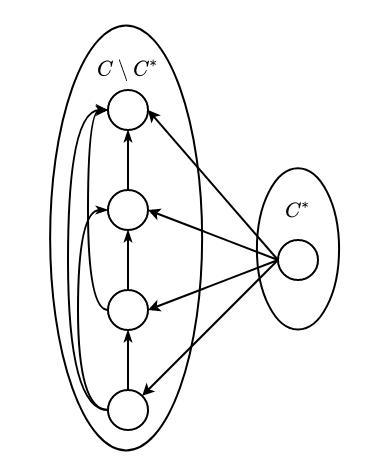}
    \caption{A hard instance when fewer than $m-1$ pairwise comparisons are elicited. }
    \label{fig:disconnected}
\end{figure}

In fact, the same limitation applies even if we allow randomization, either during the elicitation or the winner-determination phase. Importantly, we will show that $m-1$ edges from $T$ suffice to obtain near-optimal distortion. To this end, we will employ a powerful technical lemma by Kempe, proved via an LP-duality argument.

\begin{lemma}[\cite{DBLP:conf/aaai/000120a}]
    \label{lemma:propagation_error}
Let $a_1, a_2, \dots a_{\ell}$ be a sequence of distinct candidates such that for every $i = 2, \dots, \ell$ at least half of the agents prefer candidate $a_{i-1}$ over candidate $a_i$. Then, $\SC(a_1) \leq (2\ell - 1) \SC(a_{\ell})$.
\end{lemma}

Armed with this important lemma we introduce the $\DR$ mechanism, which determines a winning candidate with access only to a pairwise comparison oracle; namely, $\mathfrak{O}$ takes as input two distinct candidates $a, b \in C$ and returns the \emph{losing} candidate based on the voters' preferences (recall that in case of a tie the oracle returns an arbitrary candidate).

\begin{mdframed}[
    linecolor=black,
    linewidth=1pt,
    roundcorner=8pt,
    backgroundcolor=custom_color,
    userdefinedwidth=\textwidth,
]
$\textsc{DominationRoot}$ Mechanism\\
\textbf{Input}: Set of candidates $C$, Pairwise comparison oracle $\mathfrak{O}$ \\
\textbf{Output}: Winner $w \in C$
\begin{enumerate}
    \item[\textit{1.}] Initialize $S := C$
    \item[\textit{2.}] Construct arbitrarily a set $\Pi$ of $\lfloor S/2 \rfloor$ pairings from $S$
    \item[\textit{3.}] For every $\{a,b\} \in \Pi$ remove $\mathfrak{O}(a, b)$ from $S$
    \item[\textit{4.}] If $|S| = 1$ \textbf{return} $w \in S$; otherwise, continue from step $2$
\end{enumerate}
\end{mdframed}

We refer to \Cref{fig:domination_root} for an illustration of $\DR$. The analysis of this mechanism boils down to the following simple claims:

\begin{figure}[!ht]
    \centering
    \includegraphics[scale=0.35]{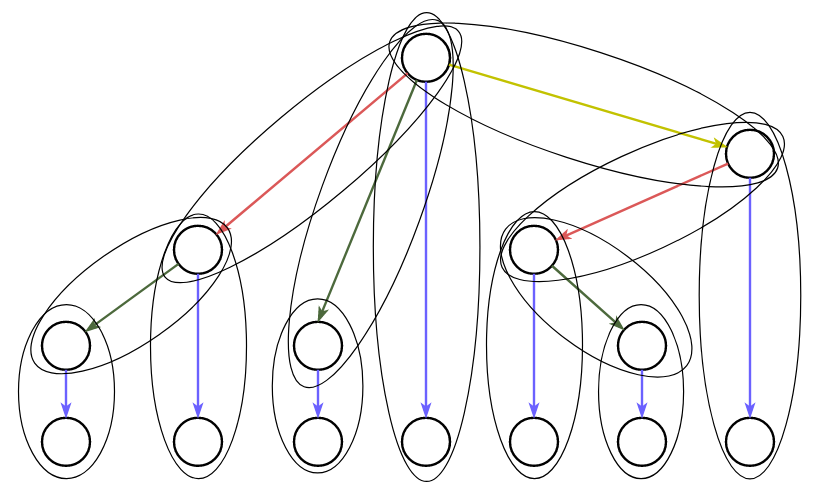}
    \caption{An implementation of $\DR$ for $m=14$ candidates. We have highlighted with different colors pairwise comparisons that correspond to different rounds of the mechanism; also notice that the ``height'' of every candidate indicates the order of elimination.}
    \label{fig:domination_root}
\end{figure}

\begin{claim}
    \label{claim:edges_elicited}
$\DR$ elicits exactly $m-1$ edges from $T$.
\end{claim}
\begin{proof}
The claim follows given that for every elicited edge we remove a candidate for the rest of the mechanism, until only a single candidate survives.
\end{proof}

\begin{claim}
    \label{claim:logm_paths}
    $\DR$ returns a candidate $w$ which can reach every other node in $T$ in paths of length at most $\lceil \log m \rceil$.
\end{claim}

\begin{proof}
Consider the partition of candidates $C_1, \dots, C_{r}$ such that $C_i$ contains the candidates who were eliminated during the $i$-th round for $i \in \{1, 2, \dots, r-1\}$, and $C_r = \{w\}$. Observe that every candidate $a \in C_i$ (with $i \in \{1, 2, \dots, r-1\}$) was pairwise-defeated by some candidate in $C_j$ for $j > i$; thus, the claim follows inductively since $r = \lceil \log m \rceil$.
\end{proof}

\begin{theorem}
    \label{theorem:pairwise_comparisons}
    $\DR$ elicits only $m-1$ edges from $T$ and guarantees distortion at most $2 \lceil \log m \rceil + 1$.
\end{theorem}

\begin{proof}
The theorem follows directly from \Cref{claim:edges_elicited}, \Cref{claim:logm_paths}, and \Cref{lemma:propagation_error}.
\end{proof}

This theorem along with \Cref{proposition:disconnected} imply a remarkable gap depending on whether the mechanism is able to elicit at least $m-1$ pairwise comparisons. We also provide a matching lower bound for $\DR$.

\begin{proposition}
    \label{proposition:DR-lower_bound}
    There exist instances for which $\DR$ yields distortion at least $2 \log m + 1$.
\end{proposition}

\begin{proof}
We will first show that the bound established in \Cref{lemma:propagation_error} is tight. Indeed, consider a set of $\ell$ candidates $\{1, 2, \dots, \ell\}$ and two voters (the instance directly extends to an arbitrary even number of voters) positioned according to the pattern of \Cref{fig:sub1}. Observe that---at least under some tie-breaking---candidate $i$ pairwise-defeats candidate $i-1$ for $i = 2, 3, \dots, \ell$. Moreover, it follows that $\SC(i) = 2i - 1$, for all $i$, implying that $\SC(\ell)/\SC(1) = 2\ell - 1$, as desired.

Now consider $m$ candidates such that $m$ is a power of $2$. We first consider $\ell = \log m + 1$ candidates positioned according to our previous argument (\Cref{fig:sub1}); the rest of the candidates are located arbitrarily far from the voters. It is easy to see that there exists a sequence of pairings (\Cref{fig:sub2}) such that $c_{\ell}$ will be declared victorious, leading to a distortion of $2 \log m + 1$ by virtue of our previous argument.
\end{proof}

\begin{figure}[!ht]
\centering
\begin{subfigure}{.5\textwidth}
  \centering
  \includegraphics[scale=0.3]{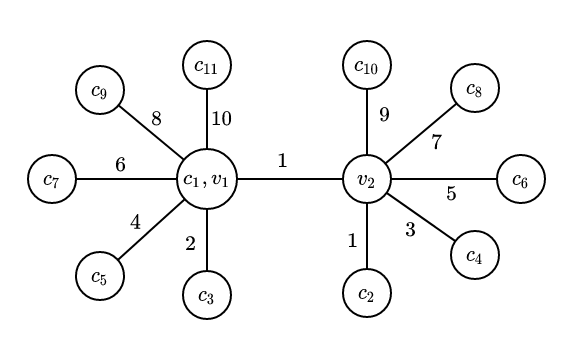}
  \caption{A metric embedding of voters and candida-\\tes establishing that \Cref{lemma:propagation_error} is tight.}
  \label{fig:sub1}
\end{subfigure}%
\begin{subfigure}{.5\textwidth}
  \centering
  \includegraphics[scale=0.3]{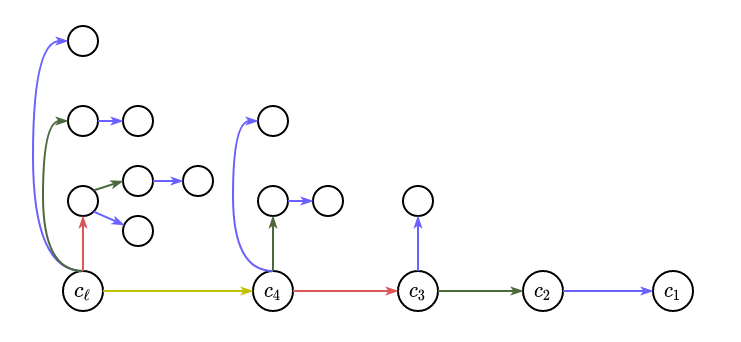}
  \caption{A sequence of pairings such that $c_{\ell}$ emerges victorious. We have highlighted with different colors pairings that correspond to different rounds.}
  \label{fig:sub2}
\end{subfigure}
\label{fig:test}
\end{figure}

\section{Distortion of Deterministic Rules under Incomplete Orders}

We commence this section with another useful lemma by Kempe \cite{DBLP:conf/aaai/000120a}.

\begin{lemma}[\cite{DBLP:conf/aaai/000120a}]
    \label{lemma:approx-transitive}
Consider three distinct candidates $w, y, x \in C$ such that $|wy| \geq \alpha n > 0$ and $|yx| \geq \alpha n > 0$. Then, 

\begin{equation}
    \frac{\SC(w)}{\SC(x)} \leq \frac{2}{\alpha} + 1.
\end{equation}
\end{lemma}

In particular, notice that if $w$ is the winner in Copeland's rule it follows that for any candidate $x$ there exists some other candidate $y$ such that $w$ pairwise-defeats $y$ and $y$ pairwise-defeats $x$ \cite{10.2307/41105842}; thus, applying \Cref{lemma:approx-transitive} for $\alpha = 1/2$ implies that the winner in Copeland's rule has distortion upper-bounded by $5$, as initially articulated in \cite{DBLP:conf/aaai/AnshelevichBP15}.

As a warm-up, we first employ this lemma to characterize the distortion when for all pairs of candidates at least a small fraction of voters has provided their pairwise preferences. We stress that all of our upper bounds are attainable by the $\minimax$ rule, but nonetheless our proofs are constructive in the sense that we provide (efficient) mechanisms which obtain the desired bounds.

\begin{proposition}
    \label{proposition:balanced}
    Consider an election $\mathcal{E} = (V, C, \mathcal{P})$ such that for every pair of distinct candidates $a, b \in C$, it holds that $\sum_{i=1}^n \mathbbm{1} \left\{ (a,b) \in \mathcal{P}_i \lor (b,a) \in \mathcal{P}_i \right\} \geq \alpha \cdot n$. Then, there exists a voting rule which obtains distortion at most $4/\alpha + 1$.
\end{proposition}

\begin{proof}
Consider a complete, weighted and directed graph $G = (C, E, w)$ such that 

\begin{equation}
 w_{a,b} = \frac{\sum_{i=1}^n \mathbbm{1} \{(a,b) \in \mathcal{P}_i \}}{n}.
\end{equation}

In words, $w_{a,b}$ represents the fraction of voters who \emph{certainly} prefer $a$ to $b$; observe that if we had the complete rankings it would follow that $w_{a,b} + w_{b,a} = 1$, but here we can only say that $w_{a,b} + w_{b, a} \leq 1$. Moreover, by assumption we know that $w_{a,b} + w_{b,a} \geq \alpha$, implying that $w_{a,b} \geq \alpha/2$ or $w_{b,a} \geq \alpha/2$. With that in mind, we construct from $G$ an unweighted and directed graph $\widehat{G} = (C, \widehat{E})$ according to the following threshold rule: $(a,b) \in \widehat{E} \iff w_{a,b} \geq \alpha/2$.
We argued that our assumption implies that $(a,b) \in \widehat{E} \lor (b,a) \in \widehat{E}$. As a result, we can deduce that $\widehat{G}$ contains as a subgraph a tournament; thus, there exists a \emph{king} vertex $w$ so that every node $a \in C$ is reachable from $w$ in at most $2$ steps, and our claim follows from \Cref{lemma:approx-transitive}.
\end{proof}

We remark that this upper bound is tight up to constant factors, at least for certain instances. Indeed, if we only have an $\alpha$ fraction of the votes in the presence of $2$ candidates it is easy to show an $\Omega(1/\alpha)$ lower bound for any mechanism, even if we allow randomization. Interestingly, \Cref{proposition:balanced} suggests one possible preference elicitation strategy: collect the information about the preferences in a ``balanced'' manner.

\subsection{Missing Voters}

Consider an election $\mathcal{E} = (V, C, \mathcal{P})$ and a mechanism which has access to the votes of only a subset $V \setminus Q$ of voters, where $Q \subset V$ is the set of \emph{missing voters} such that $|Q| = \epsilon \cdot n$. If the mechanism performs well on $V \setminus Q$ can we characterize the distortion over the entire set of voters as $\epsilon$ increases? Observe that this setting is tantamount to $\mathcal{P}_i = \emptyset$ for all $i \in Q$. In the following theorem we provide a sharp bound:

\begin{theorem}
    \label{theorem:missing_voters}
Consider a mechanism with distortion at most $\ell$ w.r.t. an arbitrary subset with $(1-\epsilon)$ fraction of all the voters, for some $\epsilon \in (0,1)$. Then, the distortion of the mechanism w.r.t. the entire population is upper-bounded by 

\begin{equation}
    \ell + \frac{\epsilon}{1 - \epsilon} (\ell + 1).
\end{equation}
\end{theorem}

\begin{proof}
Consider a candidate $b \in C$ with distortion at most $\ell$ w.r.t. the agents in $V \setminus Q$. Moreover, consider some arbitrary candidate $a \in C$, and let $S_b = \sum_{i \in V \setminus Q} d(v_i, c_b)$, and $S_a = \sum_{i \in V \setminus Q} d(v_i, c_a)$; observe that (by assumption) $S_b/S_a \leq \ell$. Our analysis will distinguish between the following two cases:

\paragraph{Case I:} $S_b \geq S_a > 0$.\footnote{The case where $S_a = 0$ can be trivially handled. Indeed, it implies that $S_b \leq \ell \times S_a = 0$, which in turn yields that $d(c_a, c_b) = 0$; thus, $\SC(a) = \SC(b)$.} Then, for all $i \in Q$ it follows that 

\begin{equation}
    S_b d(v_i, c_a) + S_a d(c_a, c_b) \geq S_a (d(v_i, c_a) + d(c_a, c_b)) \geq S_a d(v_i, c_b), 
\end{equation}
and hence,

\begin{equation}
    S_a d(v_i, c_b) \leq S_a d(c_a, c_b) + S_b d(v_i, c_a) + d(c_a, c_b) d(v_i, c_a);
\end{equation}
summing over all $i \in Q$ gives

\begin{align}
    \label{eq:missing-caseI-1}
    S_a \sum_{i \in Q} d(v_i, c_b) &\leq |Q| S_a d(c_a, c_b) + S_b \sum_{i \in Q} d(v_i, c_a) + d(c_a, c_b) \sum_{i \in Q} d(v_i, c_a) \notag \\
     &\leq |Q| S_a d(c_a, c_b) + S_b \sum_{i \in Q} d(v_i, c_a) + |Q| d(c_a, c_b) \sum_{i \in Q} d(v_i, c_a).
\end{align}
Moreover, observe that 

\begin{equation}
    \label{eq:missing-caseI-2}
    \eqref{eq:missing-caseI-1} \iff \frac{S_b + \sum_{i \in Q} d(v_i, c_b)}{S_a + \sum_{i \in Q}d(v_i, c_a)} \leq \frac{S_b + |Q| d(c_a, c_b)}{S_a}.
\end{equation}
Next, we have that $d(c_a, c_b) \leq d(v_i, c_a) + d(v_i, c_b), \forall i$; summing over all $i \in V \setminus Q$ implies that $(n - |Q|) d(c_a, c_b) \leq S_a + S_b \leq (\ell + 1) S_a$. Therefore, along with \eqref{eq:missing-caseI-2} we obtain that 

\begin{equation}
    \frac{\SC(b)}{\SC(a)} \leq \ell + \frac{|Q|}{n - |Q|} (\ell + 1) = \ell + \frac{\epsilon}{1 - \epsilon} (\ell + 1).
\end{equation}

\paragraph{Case II:} $S_b < S_a$. In this case we can simply observe that 

\begin{equation}
    \frac{\SC(b)}{\SC(a)} \leq \frac{S_b + \sum_{i \in Q} d(v_i, c_a) + |Q| d(c_a, c_b)}{S_a + \sum_{i \in Q} d(v_i, c_a)} \leq 1 + |Q| \frac{d(c_a, c_b)}{S_a}.
\end{equation}
Thus, the proof follows given that $(n - |Q|) d(c_a, c_b) \leq S_a + S_b < 2 S_a$.
\end{proof}

A few remarks are in order. First of all, if all the voters in the set $V \setminus Q$ had provided their entire rankings we would derive a similar result via \Cref{proposition:balanced}, but nonetheless \Cref{theorem:missing_voters} gives a more precise characterization. Indeed, we claim that the derived bound is tight. For example, consider an instance on the line with only two candidates $a, b$, so that every candidate receives half of the votes among the voters in $V \setminus Q$; assume without loss of generality that $a$ is selected as the winning candidate, having distortion $3$ w.r.t. the voters in $V \setminus Q$. However, we have to accept that $(1 - \epsilon)/2$ fraction of the voters could reside in the midpoint $(c_a + c_b)/2$, while the rest of the agents could all lie in $c_b$; thus, the distortion of candidate $a$ is 

\begin{equation}
    \frac{\SC(a)}{\SC(b)} = 
    \frac{\frac{1-\epsilon}{2} \frac{d(c_a, c_b)}{2} + \frac{1-\epsilon}{2} d(c_a, c_b) + \epsilon d(c_a, c_b)}{\frac{1-\epsilon}{2} \frac{d(c_a, c_b)}{2}} = 3 + 4 \frac{\epsilon}{1 - \epsilon},
\end{equation}
which matches our derived bound in \Cref{theorem:missing_voters}.

\subsection{Top Preferences}

In this subsection we investigate how the distortion increases when every voter provides only her $k$-top preferences, for some parameter $k \in [m]$. It should be noted that the two extreme cases are well understood. Specifically, when $k = m$ the mechanism has access to the entire rankings and we know that any deterministic mechanism has distortion at least $3$, which is also the upper bound established in \cite{DBLP:conf/focs/Gkatzelis0020}. On the other end of the spectrum, when $k = 1$ the plurality rule---which incidentally is the optimal deterministic mechanism when only the top preference is given---yields distortion at most $2m - 1$ \cite{DBLP:conf/aaai/AnshelevichBP15}. Consequently, the question is quantify the decay of distortion as we gradually increase $k$. We commence by reminding a lower bound given by Kempe \cite{DBLP:conf/aaai/000120b}:

\begin{proposition}
    Any deterministic mechanism which elicits only the $k$-top preferences from every voter out of the $m$ alternatives has distortion $\Omega(m/k)$.
\end{proposition}

More precisely, the best lower bound is $2 m/k$, ignoring some additive constant factors; for completeness we provide a proof in \Cref{appendix:lower_bound}. In the following theorem we come closer to matching this lower bound.

\begin{theorem}
    \label{theorem:k_top-upper_bound}
    There exists a deterministic mechanism which elicits only the $k$-top preferences from every voter out of $m$ candidates and has distortion at most $6 m/k + 1$.
\end{theorem}

Before we proceed with the proof it is important to point out that having only the $k$-top preferences is not subsumed by our previous consideration in \Cref{proposition:balanced}; e.g., even if $k = m-2$ there could be two candidates which lie on the last two positions of every voter's list, and hence, it is impossible to know which one is mostly preferred among the voters. 

\begin{proof}[Proof of \Cref{theorem:k_top-upper_bound}]
Let $\mathcal{L}_i$ be the set with the $k$-top preferences of voter $i$. For a candidate $a \in C$ we let

\begin{equation}
\mathcal{V}_a = \frac{\sum_{i=1}^n \mathbbm{1}\{ a \in \mathcal{L}_i \}}{n};
\end{equation}
i.e. the fraction of voters for which $a$ is among the $k$-top preferences. Notice that $\sum_{a \in C} \mathcal{V}_a = k$, and hence, by the pigeonhole principle there exists some candidate $x$ such that $\mathcal{V}_x \geq k/m$. Similarly to \Cref{proposition:balanced} we consider the weighted, complete and directed graph $G = (C, E, w)$, so that 

\begin{equation}
    w_{a,b} = \frac{\sum_{i=1}^n \mathbbm{1} \{(a,b) \in \mathcal{P}_i \}}{n}.
\end{equation}
Moreover, based on $G$ we construct the unweighted and directed graph $\widehat{G} = (C, \widehat{E})$, so that $(a, b) \in \widehat{E} \iff w_{a,b} \geq k/(3m)$; the constant $1/3$ in the threshold was selected as the largest number which makes the following argument work. In particular, we will show that $\widehat{G}$ has a king vertex, and then the theorem will follow by virtue of \Cref{lemma:approx-transitive}. 

Let $C' = \{ a \in C : \exists b \in C \setminus \{a\}. (a,b) \notin \widehat{E} \land (b,a)\notin \widehat{E} \}$ and $C^* = C \setminus C'$. Observe that the \emph{induced} subgraph on $C^*$ contains as a subgraph a tournament, and as such, it has a king vertex $w \in C^*$ (we will argue very shortly that indeed $C^* \neq \emptyset$). As a result, if $C' = \emptyset$ the theorem follows.

In the contrary case notice that $C'$ contains at least two (distinct) nodes; let $a, b \in C'$ be such that $(a,b) \notin \widehat{E} \land (b,a) \notin \widehat{E}$. An important observation is that $\mathcal{V}_a \leq 2k/(3m)$ and $\mathcal{V}_b \leq 2k/(3m)$. Indeed, for the sake of contradiction let us assume that $\mathcal{V}_a > 2k/(3m)$. Given that $(a,b) \notin \widehat{E}$ we can infer that $b$ is preferred over $a$ in at least a $k/(3m)$ fraction of the voters; however, this would imply that $(b,a) \in \widehat{E}$, which is a contradiction. Similarly, we can show that $\mathcal{V}_b \leq 2k/(3m)$. Consequently, $x$ cannot belong in the set $C'$, where recall that $x$ is a candidate for which $\mathcal{V}_x \geq k/m$, verifying that $C^* \neq \emptyset$.

Next, it is easy to see that for all $a \in C', (x,a) \in \widehat{E}$; this follows since $\mathcal{V}_a \leq 2k/(3m)$ for all $a \in C'$ while $\mathcal{V}_x \geq k/m$. As a result, if $x = w$ or if there exists the edge $(w, x) \in \widehat{E}$ then $w$ can reach every node in at most $2$ steps, and the theorem follows. Otherwise, it follows that there exists a path of length $2$ from $w$ to $x$ since $w$ is a king vertex in the induced subgraph on $C^*$ and $x \in C^*$. We shall distinguish between two cases.

First, assume that for all $z \in C_1^*, (x, z) \in \widehat{E}$, where $C_1^*$ is the subset of $C^*$ which is reachable from $w$ via a single edge. Then, given that we have assumed that $(w,x) \notin \widehat{E}$ and the induced graph on $C^*$ is a tournament, it follows that $(x,w) \in \widehat{E}$ and subsequently $x$ can reach every node in $C$ in paths of length at most $2$, as desired.

Finally, assume that there exists some $y \in C_1^*$ such that $(x,y) \notin \widehat{E}$. This implies that $y$ is preferred over $x$ in at least a $2k/(3m)$ fraction of the voters. If for every candidate $a \in C'$ it holds that $(w,a) \in \widehat{E}$ or $(z, a) \in \widehat{E}$ for some $z \in C_1^*$, we can conclude that $w$ can reach every node in $\widehat{G}$ in at most $2$ steps, again reaching the desired conclusion. On the other hand, assume that there exists $b \in C'$ such that $(w,b) \notin \widehat{E}$ and $(z,b) \notin \widehat{E}$ for all $z \in C_1^*$. By the definition of the set $C^*$ we can infer that $(b,w) \in \widehat{E}$ and $(b,z) \in \widehat{E}$ for all $z \in C_1^*$. Moreover, we know that from all of the votes candidate $y$ received candidate $b$ was below in at most a $k/(3m)$ fraction (over all the voters); otherwise it would follow that $(y,b) \in \widehat{E}$. As a result, since $y$ is preferred over $x$ in at least a $2k/(3m)$ fraction of the voters we can conclude that $(b,x) \in \widehat{E}$, in turn implying that $b$ can reach every node in $\widehat{G}$ in paths of length at most $2$, concluding the proof.
\end{proof}

\begin{figure}[!ht]
    \centering
    \includegraphics[scale=0.35]{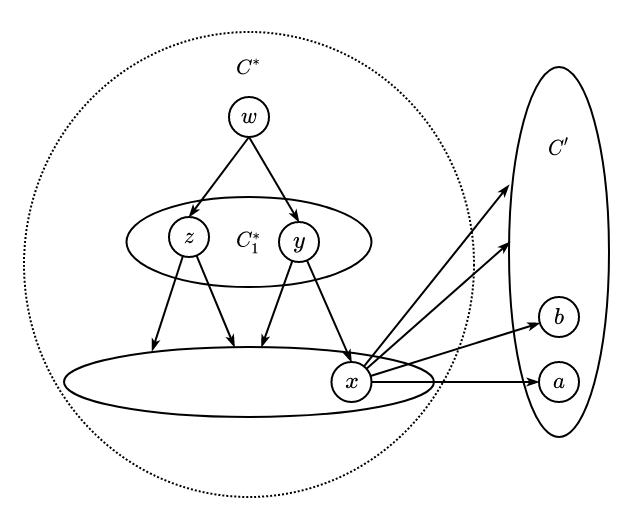}
    \caption{The anatomy of our proof for \Cref{theorem:k_top-upper_bound}. The set of candidates is partitioned into a ``good'' set $C^*$ and a ``bad'' set $C'$; $C^*$ has a king vertex $w$, and we can essentially apply the reasoning of \Cref{proposition:balanced}. A key observation is that $C'$ is always dominated by some node in $C^*$, namely $x$. }
    \label{fig:k_top-proof}
\end{figure}

Notably, we have shown that if $k = \gamma \cdot m$ for some $\gamma \in (0,1)$, the distortion is at most $6/\gamma + 1$. Our analysis substantially improves over the previous best-known bound which was $12m/k$ \cite{DBLP:conf/aaai/000120b,DBLP:conf/aaai/000120a}, but nonetheless there is still a gap between the aforementioned lower bound. Before we conclude this section, we explain how one can further improve upon the bound obtained in \Cref{theorem:k_top-upper_bound}.

\begin{conjecture}
    \label{conjecture:extension}
If we assume that every agent provides her $k$-top preferences for some $k \in [m]$, there is a candidate $a \in C$ and a subset $S \subseteq V$ such that 
\begin{itemize}
    \item There exists a perfect matching $M : S \mapsto S$ in the integral domination graph of $a$ (see \Cref{definition:integral_domination_graph} in the next section);
    \item $|S| \geq n \times k/m$.
\end{itemize}
\end{conjecture}

When $k = m$ this conjecture was shown to be true by Gkatzelis et al. \cite{DBLP:conf/focs/Gkatzelis0020}. On the other end of the spectrum, when $k = 1$ it is easy to verify that the plurality winner establishes this conjecture. 

\begin{proposition}
    \label{proposition:conditional}
    If \Cref{conjecture:extension} holds, then there exists a deterministic mechanism which elicits only the $k$-top preferences and yields distortion at most $4 m/k - 1$.
\end{proposition}

\begin{proof}
Let $a \in C$ be the candidate which satisfies \Cref{conjecture:extension}. Then, it follows that $a$ yields distortion at most $3$ w.r.t. the voters in the set $S$ \cite{DBLP:conf/focs/Gkatzelis0020}. As a result, \Cref{theorem:missing_voters} implies that the distortion of $a$ is upper-bounded by 

\begin{equation}
    3 + 4 \frac{n - |S|}{|S|} \leq \frac{4m}{k} - 1.
\end{equation}

\end{proof}

\section{Randomized Preference Elicitation \& Sampling}

Previously we characterized the distortion when only a deterministically (and potentially adversarially) selected subset of voters has provided information to the mechanism. This raises the question of bounding the distortion when the mechanism elicits information from only a small \emph{random sample} of voters. Here a single sample corresponds to the \emph{entire} ranking of a voter. We stress that randomization is only allowed during the preference elicitation process; for any given random sample as input the mechanism has to select a candidate \emph{deterministically}. We commence this section with a simple lower bound, which essentially follows from a standard result by Canetti et al. \cite{DBLP:journals/ipl/CanettiEG95}.

\begin{proposition}
    \label{proposition:lower_bound-sample_complexity}
    Any mechanism which yields distortion at most $3 + \epsilon$ with probability at least $1 - \delta$ requires $\Omega(\log(1/\delta)/\epsilon^2)$ samples, even for $m=2$.
\end{proposition}

\begin{proof}
Consider two candidates $a, b$, and assume that exactly $(1 - \epsilon)/2$ fraction of the voters prefer candidate $a$. It is easy to verify that $a$ yields distortion strictly larger than $3 + \epsilon$; thus, any mechanism with distortion at most $3 + \epsilon$ has to return candidate $b$. However, we know from \cite{DBLP:journals/ipl/CanettiEG95} that the \emph{winner determination} problem with \emph{margin} $\epsilon$ requires $\Omega(\log(1/\delta)/\epsilon^2)$ samples, concluding the proof.
\end{proof}

\subsection{Approximating Copeland}

Our main result in this subsection is the following:

\begin{theorem}
    \label{theorem:sample-complexity}
    For any $\epsilon \in (0,4]$ and $\delta \in (0,1)$ there exists a mechanism which takes a sample of size $c = \Theta(\log(m/\delta)/\epsilon^2)$ voters and yields at most $5 + \epsilon$ distortion with probability at least $1 - \delta$.
\end{theorem}

In particular, the proof essentially analyzes a \emph{sampling approximation} of Copeland's rule, which recall that yields at most $5$ distortion when the entire input is available \cite{DBLP:conf/aaai/AnshelevichBP15}. As a result, it follows that $\widetilde{\Theta}(m/\epsilon^2)$ bits of information (in total) suffice to yield $5 + \epsilon$ distortion with high probability, where the notation $\widetilde{\Theta}(\cdot)$ suppresses poly-logarithmic factors. Before we proceed with the proof we state the following standard fact:

\begin{lemma}[Chernoff-Hoeffding Bound]
    \label{lemma:chernoff}
    Let $\{X_1, X_2, \dots, X_c\}$ be a set of i.i.d. random variables with $X_i \sim \text{Bern}(p)$ and $X_{\mu} = (X_1 + X_2 + \dots + X_c)/c$; then,
    \begin{equation}
        \mathbb{P}(|X_{\mu} - p| \geq \epsilon) \leq 2 e^{-2 \epsilon^2 c}.
    \end{equation}
\end{lemma}

\begin{proof}[Proof of \Cref{theorem:sample-complexity}]

Consider the complete, weighted and directed graph $G = (C, E, w)$ so that $w_{a,b} = |ab|/n$. We will show how to use the random sample in order to construct a graph $\widehat{G} = (C, E, \widehat{w})$ which approximately preserves the weights of $G$ with high probability. In particular, consider some parameters $\epsilon \in (0,1/2)$ and $\delta \in (0,1)$, and take a sample $S$ of size $|S| = c = \Theta(\log(m/\delta)/\epsilon^2)$ from the set of voters $V$; for simplicity we assume that the sampling occurs \emph{with replacement} in order to guarantee independence, but the result holds even without replacement given that the dependencies are negligible (e.g., see \cite{10.1214/07-AOP384}). Now we let $\widehat{w}_{a,b} = |\{i \in S : a \succ_i b \} |/c$. \Cref{lemma:chernoff} implies that $|\widehat{w}_{a,b} - w_{a,b}| < \epsilon$ with probability at least $1 - \delta/m^2$. Thus, the union bound implies that for all distinct pairs $a, b$ we have approximately preserved the weights: $|\widehat{w}_{a,b} - w_{a,b}| < \epsilon$ with probability at least $1 - \delta$. 

From $\widehat{G}$ we construct the directed graph $T = (C, \widehat{E})$ so that $(a,b) \in \widehat{E} \iff \widehat{w}_{a,b} \geq 1/2$; if $\widehat{w}_{a,b} = \widehat{w}_{b,a}$ for some distinct candidates $a, b \in C$ we only retain one of the edges arbitrarily (this conundrum can be avoided by taking $c$ to be odd). In this way $T$ will be a tournament, and as such, there exists a candidate $w$ which can reach every node in $T$ in at most $2$ steps. Thus, for any $a \in C$ there exists some intermediate candidate $b \in C$ so that $|wb| \geq 1/2 - \epsilon$ and $|ba| \geq 1/2 - \epsilon$. As a result, \Cref{lemma:approx-transitive} implies that the distortion of $w$ is upper-bounded by $4/(1-2\epsilon) + 1 \leq 5 + 16 \epsilon$, for any $\epsilon \leq 1/4$. Finally, rescaling $\epsilon$ by a constant factor concludes the proof.
\end{proof}

\subsection{Approximating \texorpdfstring{$\pluralitymatching$}{}}

In light of \Cref{proposition:lower_bound-sample_complexity} the main question that arises is whether we can asymptotically reach the optimal distortion bound of $3$. To this end, we will analyze a sampling approximation of $\pluralitymatching$, a deterministic mechanism introduced by Gkatzelis et al. \cite{DBLP:conf/focs/Gkatzelis0020} which obtains the optimal distortion bound of $3$. To keep the exposition reasonably self-contained we recall some basic facts about $\pluralitymatching$.

\begin{definition}[\cite{DBLP:conf/focs/Gkatzelis0020}, Definition $5$]
    \label{definition:integral_domination_graph}
    For an election $\mathcal{E} = (V, C, \sigma)$ and a candidate $a \in C$, the integral domination graph of candidate $a$ is the bipartite graph $G(a) = (V, V, E_a)$, where $(i, j) \in E_a$ if and only if $a \succeq_i \topp(j)$.
\end{definition}

\begin{proposition}[\cite{DBLP:conf/focs/Gkatzelis0020}, Corollary $1$]
    \label{proposition:perfect_matching}
    There exists a candidate $a \in C$ whose integral domination graph $G(a)$ admits a perfect matching.
\end{proposition}

Before we proceed let us first introduce some notation. For this subsection it will be convenient to use numerical values in the set $\{1, 2, \dots, m\}$ to represent the candidates. We let $\Pi_j = \sum_{i \in V} \mathbbm{1} \{ \topp(i) = j \}$, i.e. the number of voters for which $j \in C$ is the top candidate. For candidate $j \in C$ we let $G(j)$ be the integral domination graph of $j$, and $M_j$ be a maximum matching in $G(j)$. In the sequel, it will be useful to ``decompose'' $M_j$ as follows. We consider the partition of $V$ into $V_j^0, V_j^1, \dots, V_j^m$ such that $V_j^k = \{ i \in V : M_j(i) = k \}$ for all $k \in [m]$, while $V_j^0$ represents the subset of voters which remained unmatched under $M_j$; see \Cref{fig:integral_dom}.

\begin{figure}[!ht]
    \centering
    \includegraphics[scale=0.4]{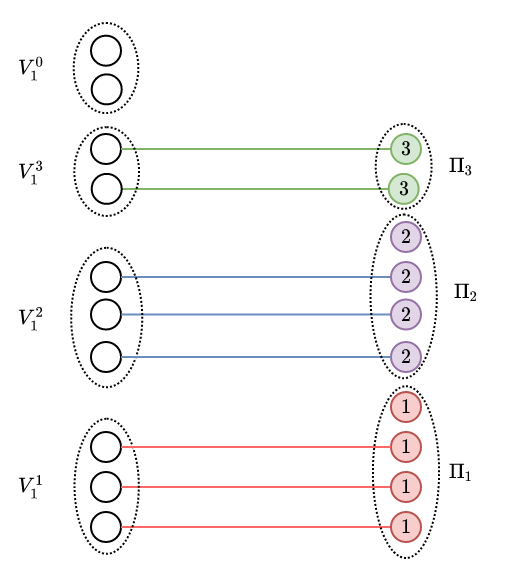}
    \caption{An example of a matching decomposition in the integral domination graph of candidate $1$.}
    \label{fig:integral_dom}
\end{figure}

Moreover, consider a set $\mathcal{S} = S_j^0 \cup S_j^1 \cup \dots \cup S_j^m$ such that $S_j^k \subseteq V_j^k$ for all $k$; we also let $c = |\mathcal{S}|$, and $\Pi_j' = c/n \times \Pi_j$. For now let us assume that $\Pi_j' \in \mathbb{N}$ for all $j$. We let $G^{\mathcal{S}}(j)$ represent the induced subgraph of $G(j)$ w.r.t. the subset $\mathcal{S} \subseteq V$ and the new plurality scores $\Pi_j'$. We start our analysis with the following observation:

\begin{observation}
    \label{observation:sampling_matching}
Assume that $\mathcal{S}$ is such that $|S_j^k|/c = |V_j^k|/n$ for all $k$. Then, if $M_j^{\mathcal{S}}$ represents the maximum matching in $G^{\mathcal{S}}(j)$, it follows that $|M_j^{\mathcal{S}}|/c = |M_j|/n$. 
\end{observation}

\begin{sproof}
First, it is clear that $|M_j^{\mathcal{S}}| \geq \sum_{k=1}^m |S_j^k| = c/n \sum_{k=1}^m |V_j^k| = c/n \times |M_j|$. Thus, it remains to show that $|M_j^{\mathcal{S}}| \leq c/n \times |M_j|$. Indeed, if we assume otherwise we can infer via an exchange argument than $M_j$ is not a maximum matching. 
\end{sproof}

Let us denote with $\Phi_j = M_j/n$; roughly speaking, we know from \cite{DBLP:conf/focs/Gkatzelis0020} that $\Phi_j$ is a good indicator of the ``quality'' of candidate $j$. Importantly, \Cref{observation:sampling_matching} tells us that we can determine $\Phi_j$ in a much smaller graph, if only we had a sampling-decomposition that satisfied the ``proportionality'' condition of the claim. Of course, determining explicitly such a decomposition makes little sense given that we do not know the sets $V_j^{k}$, but the main observation is that we can approximately satisfy the condition of \Cref{observation:sampling_matching} through sampling. It should be noted that we previously assumed that $\Pi_j' \in \mathbb{N}$, i.e. we ignored rounding errors. However, in the worst-case rounding errors can only induce an error of at most $m/c$ in the value of $\Phi_j$; thus, we remark that our subsequent selection of $c$ will be such that this error will be innocuous, in the sense that it will be subsumed by the ``sampling error'' (see \Cref{lemma:approx_flow}). Before we proceed, recall that for $\mathbf{p}, \widehat{\mathbf{p}} \in \Delta([k])$,

\begin{equation}
\tv(\mathbf{p}, \widehat{\mathbf{p}}) \define \sup_{S \subseteq [k]} |\mathbf{p}(S) - \widehat{\mathbf{p}}(S)| = \frac{1}{2} || \mathbf{p} - \widehat{\mathbf{p}} ||_1,
\end{equation}
where $||\cdot||_1$ represents the $\ell_1$ norm. In this context, we will use the following standard fact:

\begin{lemma}[\cite{canonne2020short}]
    \label{lemma:tv}
Consider a discrete distribution $\mathbf{p} \in \Delta([k])$ and let $\widehat{\mathbf{p}}$ be the empirical distribution derived from $N$ independent samples. For any $\epsilon > 0$ and $\delta \in (0,1)$, if $N = \Theta((k + \log(1/\delta))/\epsilon^2)$ it follows that $\tv(\mathbf{p}, \widehat{\mathbf{p}}) \leq \epsilon$ with probability at least $1 - \delta$.
\end{lemma}

As a result, if we draw a set $\mathcal{S}$ with $|\mathcal{S}| = c = \Theta((m + \log(1/\delta))/\epsilon^2)$ samples (\emph{without} replacement\footnote{Although the samples are not independent since we are not replacing them, observe that the induced bias is negligible for $n$ substantially larger than $m$.}) we can guarantee that 

\begin{equation}
    \sum_{k=0}^m \left| \frac{|S_j^k|}{c} - \frac{|V_j^k|}{n} \right| \leq 2 \epsilon;
\end{equation}

\begin{equation}
    \sum_{k=1}^m \left| \frac{\widehat{\Pi}_k}{c} - \frac{\Pi_k}{n} \right| \leq 2 \epsilon,
\end{equation}
where $S_j^k$ represents the subset of $\mathcal{S}$ which intersects $V_j^k$, and $\widehat{\Pi}_k$ is the empirical plurality score of candidate $k$. Thus, the following lemma follows directly from \Cref{observation:sampling_matching} and \Cref{lemma:tv}.

\begin{lemma}
    \label{lemma:approx_flow}
Let $\widehat{\Phi}_j = |\widehat{M}_j|/c$, where $\widehat{M}_j$ is the maximum matching in the graph $G^{\mathcal{S}}(j)$. Then, if $|\mathcal{S}| = \Theta((m + \log(1/\delta))/\epsilon^2)$ for some $\epsilon, \delta \in (0,1)$, it follows that $(1 - \epsilon) \Phi_j \leq \widehat{\Phi}_j \leq (1 + \epsilon) \Phi_j$ with probability at least $1 - \delta$.
\end{lemma}

\begin{theorem}
    \label{theorem:plurality_matching}
    For any $\epsilon \in (0, 4]$ and $\delta \in (0,1)$ there exists a mechanism which takes a sample of size $\Theta((m + \log(m/\delta))/\epsilon^2)$ voters and yields distortion at most $3 + \epsilon$ with probability at least $1 - \delta$.
\end{theorem}

\begin{proof}
Fix some $\epsilon \in (0, 1/4)$ and $\delta \in (0,1)$. If we draw $\Theta((m + \log(m/\delta))/\epsilon^2)$ samples \Cref{lemma:approx_flow} along with the union bound imply that $(1-\epsilon) \Phi_j \leq \widehat{\Phi}_j \leq (1 + \epsilon) \Phi_j$ for all $j \in [m]$, with probability at least $1 - \delta$, where $\widehat{\Phi}_j$ is defined as in \Cref{lemma:approx_flow}. Conditioned on the success of this event, let $w = \argmax_{j \in C} \widehat{\Phi}_j$. \Cref{proposition:perfect_matching} implies that there exists some candidate $x$ for which $\Phi_x = 1$; hence, we know that $\widehat{\Phi}_w \geq \widehat{\Phi}_x \geq 1 - \epsilon$, in turn implying that $\Phi_w \geq (1-\epsilon)/(1+\epsilon) \geq 1 - 2 \epsilon$ (\Cref{lemma:approx_flow}). As a result, it follows that there exists a subset of voters $V'$ for which there exists a perfect matching in the integral domination graph $G(w)$, with $|V'| \geq n(1 - 4 \epsilon)$. Thus, it follows that for the subset of voters in $V'$ candidate $w$ yields distortion at most $3$ (see \cite{DBLP:conf/focs/Gkatzelis0020}), and \Cref{theorem:missing_voters} leads to the desired conclusion.
\end{proof}

\section{Experiments}

\subsection{Synthetic Data}

Here we present several synthetic experiments which illustrate the degradation of distortion under missing information. We will employ the $\minimax$ rule, which recall is instance-optimal under any given preferences $\mathcal{P}$ (\Cref{theorem:instance_opt}). Also note that the $\minimax$ rule returns the exact distortion of \emph{every} candidate w.r.t. the given preferences. The induced LPs in $\minimax$ rule are solved via the \texttt{Gurobi} software \cite{gurobi}.

\subsubsection{Top Preferences}

First we assume that every voter provides only her $k$-top preferences to the mechanism, and we illustrate the decay of distortion while $k$ gradually increases from $1$ to $m$. Specifically, our experiments are conducted for $n = 50$ voters with preferences sampled from a uniform distribution over the space of permutations. In the parlor of social choice this probabilistic model is referred to as \emph{impartial culture},\footnote{To put it differently, impartial culture corresponds to the \emph{Mallows} model \cite{10.1093/biomet/44.1-2.114} with unitary \emph{spread parameter}.} and arguably it is unrealistic \cite{10.2307/41106568}; nonetheless it will suffice for sketching the underlying qualitative behavior. Recall that according to \Cref{theorem:k_top-upper_bound} we expect the distortion to decay as $\mathcal{O}(m/k)$. The results are shown in \Cref{fig:k-top} for $m \in \{4, 6, 8, 10\}$; it should also be noted that for every case we presented $5$ different random realizations in order to alleviate the ``bias'' in the input. Interestingly, the ``bottom half'' of the voters' preferences appears to offer no real improvement.

\begin{figure}[!ht]
\begin{minipage}{.5\linewidth}
\centering
\subcaption{\label{main:a}\includegraphics[scale=.48]{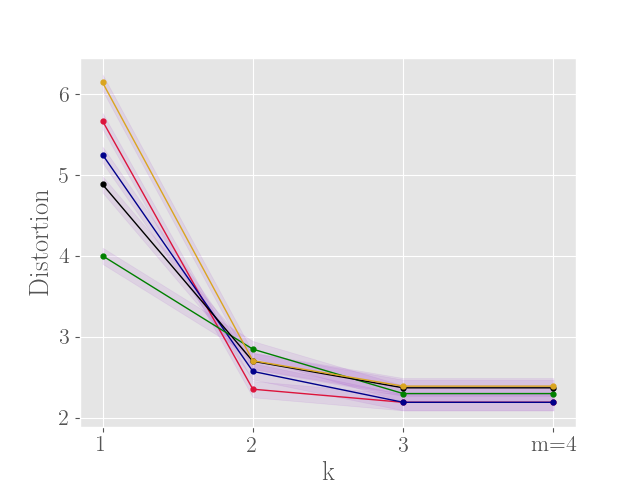}}
\end{minipage}%
\begin{minipage}{.5\linewidth}
\centering
\subcaption{\label{main:b}\includegraphics[scale=.48]{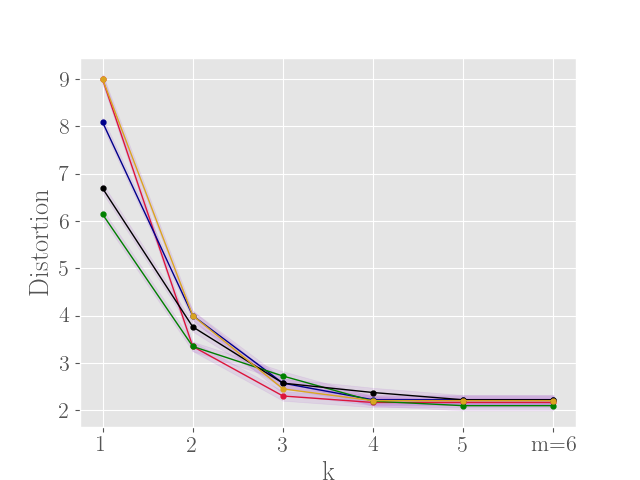}}
\end{minipage}\par\medskip
\begin{minipage}{.5\linewidth}
\centering
\subcaption{\label{main:c}\includegraphics[scale=.48]{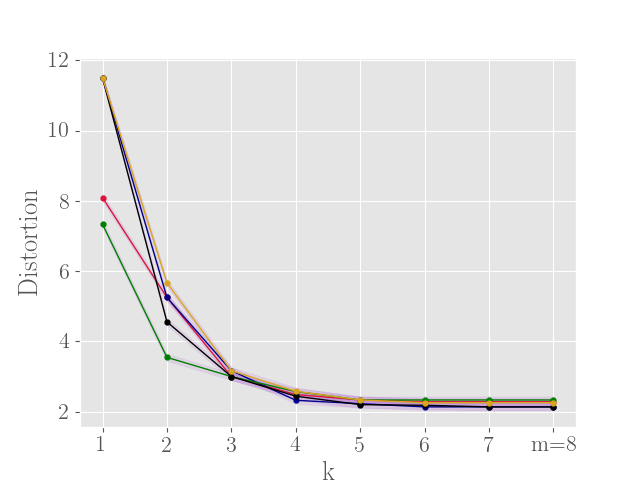}}
\end{minipage}%
\begin{minipage}{.5\linewidth}
\centering
\subcaption{\label{main:d}\includegraphics[scale=.48]{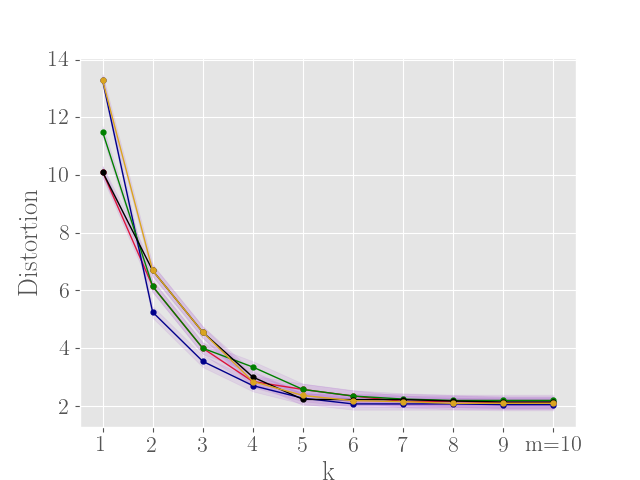}}
\end{minipage}\par\medskip
\caption{The decay of distortion as we gradually increase $k$ from $1$ to $m$.}
\label{fig:k-top}
\end{figure}

\subsubsection{Missing Voters}

Next we illustrate the decay of distortion as more voters provide information to the mechanism; here it will be assumed that (active) voters provide their total orders to the mechanism. As before, we consider $n=50$ voters from an impartial culture probabilistic model with $m \in \{3, 4, 5, 6\}$; for every case we consider $5$ different random realizations. The results are illustrated in \Cref{fig:missing_voters}. Again, the observed curves closely match the theoretical predictions of \Cref{theorem:missing_voters}. 

\begin{figure}[!ht]
\begin{minipage}{.5\linewidth}
\centering
\subcaption{\label{main:aa}\includegraphics[scale=.48]{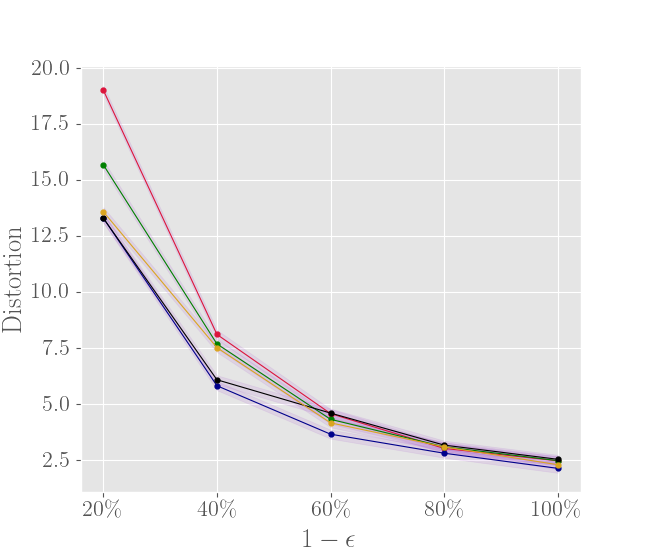}}
\end{minipage}%
\begin{minipage}{.5\linewidth}
\centering
\subcaption{\label{main:bb}\includegraphics[scale=.48]{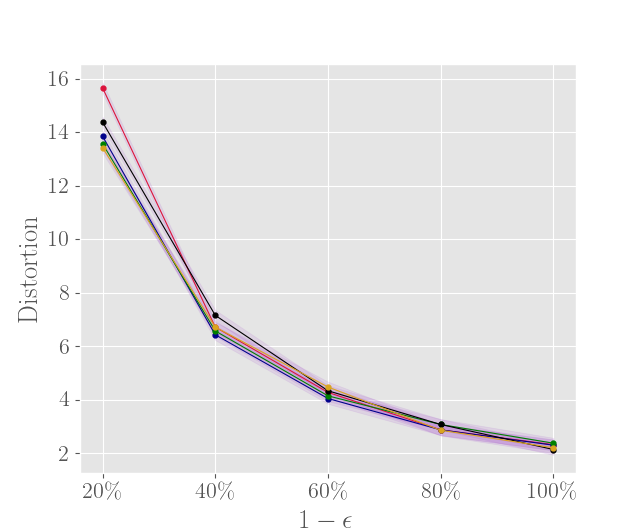}}
\end{minipage}\par\medskip
\begin{minipage}{.5\linewidth}
\centering
\subcaption{\label{main:cc}\includegraphics[scale=.48]{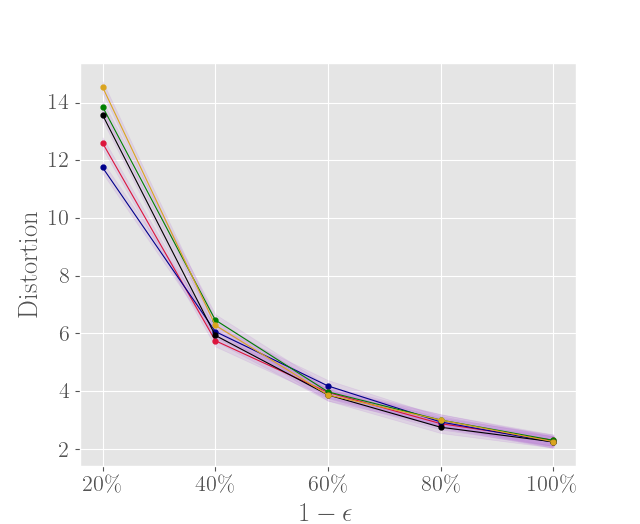}}
\end{minipage}%
\begin{minipage}{.5\linewidth}
\centering
\subcaption{\label{main:dd}\includegraphics[scale=.48]{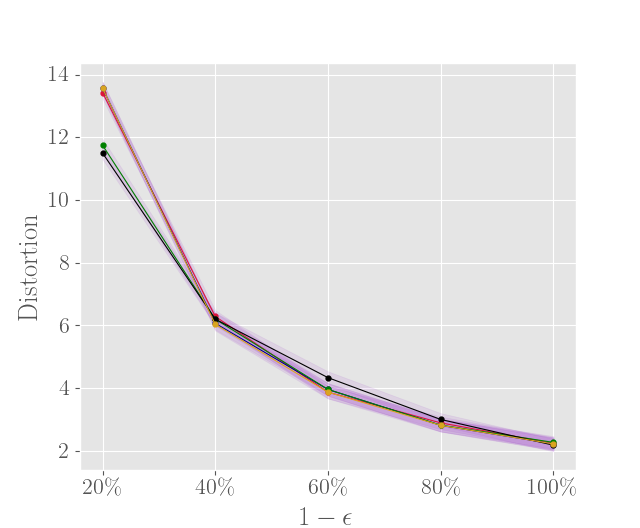}}
\end{minipage}\par\medskip
\caption{The decay of distortion as more voters (a $(1 -\epsilon)$ fraction) provide their preferences to the mechanism, for $\epsilon \in \{0.8, 0.6, 0.4, 0.2, 0\}$. The illustrated curves correspond to $m \in \{3, 4, 5, 6\}$ from top-left, top-right, bottom-left, and bottom-right respectively.}
\label{fig:missing_voters}
\end{figure}

\subsection{Real Datasets}

In this subsection we present some experiments conducted on real-life voting applications. In particular, we are interested in comparing the efficiency---measured in terms of distortion---of the voting system employed in practice with the instance-optimal mechanism, namely the $\minimax$ mechanism.

\subsubsection{Eurovision Song Contest}

Here we analyze the performance of the scoring system used in the Eurovision song contest, so let us first give a basic overview of the competition and the voting rule employed. Fist of all, we will only focus on the final stage of the competition, wherein a set of $m$ countries compete amongst each other and a set of $n$ countries---which is a strict superset of the contenders---provide their preferences over the finalists. Eurovision employs a specific \emph{positional scoring} system which works as follows. Every country assigns $12$ points to its highest preference, $10$ points to its second-highest preference, and from $8-1$ points to each of its next $8$ preferences; note that no country can vote for itself. This scoring system shall be referred to as the $\textsc{Scoring}$ rule. It should be noted that the authors in \cite{DBLP:conf/aaai/SkowronE17} quantify the distortion for some specific scoring rules (e.g. the \emph{harmonic rule}). We will make the working hypothesis that for every country the assigned scores correspond to its actual order of preferences. Nonetheless, we stress that the assigned scores of every country have been themselves obtained by \emph{preference aggregation}\footnote{For the years we are considering the scores were mainly determined by \emph{televoting}, with some few exceptions.}, and as such they are themselves subject to \emph{distortion}, but we will tacitly suppress this issue.\footnote{We refer the interested reader to the work of Filos-Ratsikas and Voudouris \cite{DBLP:journals/corr/abs-2007-06304}.}

We will focus on the competitions held between $2004$ and $2008$; during these years the number of finalists (or candidates) $m$ was $24$, with the exception of $2008$ where $25$ countries were represented in the final. We should note that for our experiments we used a \hyperlink{https://www.kaggle.com/datagraver/eurovision-song-contest-scores-19752019}{dataset} from \texttt{Kaggle}. Observe that every ``voter'' only provides its top $k=10$ preferences, while the countries which are represented in the final are $0$-\emph{decisive} (see \cite{DBLP:conf/ijcai/AnshelevichP16}). The main question that concerns us is whether the $\textsc{Scoring}$ rule employed for the competition yields very different results from the optimal $\minimax$ mechanism. Our results are summarized in \Cref{table:eurovision}, while for more detailed findings we refer to \Cref{appendix:additional_experiments}.

Perhaps surprisingly, on all occasions the winners in the two mechanisms coincide; on the other hand, there are generally substantial differences below the first position. It is also interesting to note that on all occasions the winner has a remarkably small distortion, at least compared to the theoretical bounds.

\begin{table}[h]
\tiny
\caption{Summary of our findings for the Eurovision song contests held between $2004$ and $2008$. For every year we have indicated the top three countries according to the $\minimax$ rule and the $\textsc{Scoring}$ system employed in the actual contest.} 
\centering 
\resizebox{\textwidth}{!}{\begin{tabular}{c | c c | c c | c }
\toprule
 \multicolumn{1}{c|}{\multirow{2}{*}{Year}} & \multicolumn{2}{c|}{$\minimax$ rule} & \multicolumn{2}{c|}{\textsc{Scoring} rule} & \multicolumn{1}{c}{\multirow{2}{*}{$\#$ of Countries}}
\\
& Country & Distortion & Country & Score & \\ \midrule
\multicolumn{1}{c|}{\multirow{3}{*}{$2004$}} & Ukraine & $\textbf{1.1786}$ & Ukraine & $\textbf{280}$ & \multicolumn{1}{c}{\multirow{3}{*}{$36$}} \\
& Serbia $\&$ Montenegro & $1.4444$ & Serbia $\&$ Montenegro & $263$ & \\
& Turkey & $1.4746$ & Greece & 252 & \\ \midrule
\multicolumn{1}{c|}{\multirow{3}{*}{$2005$}} & Greece & $\textbf{1.4068}$ & Greece & $\textbf{230}$ & \multicolumn{1}{c}{\multirow{3}{*}{$39$}} \\
& Switzerland & $1.4127$ & Malta & $192$ & \\
& Moldova & $1.4194$ & Romania & 158 & \\ \midrule
\multicolumn{1}{c|}{\multirow{3}{*}{$2006$}} & Finland & $\textbf{1.3000}$ & Finland & $\textbf{292}$ & \multicolumn{1}{c}{\multirow{3}{*}{$38$}} \\
& Romania & $1.4262$ & Russia & $248$ & \\
& Russia & $1.4407$ & Bosnia $\&$ Herzegovina & $229$ & \\ \midrule
\multicolumn{1}{c|}{\multirow{3}{*}{$2007$}} & Serbia & $\textbf{1.3235}$ & Serbia & $\textbf{268}$ & \multicolumn{1}{c}{\multirow{3}{*}{$42$}} \\
& Ukraine & $1.3667$ & Ukraine & $235$ & \\
& Russia & $1.5231$ & Russia & $207$ & \\ \midrule
\multicolumn{1}{c|}{\multirow{3}{*}{$2008$}} & Russia & $\textbf{1.3562}$ & Russia & $\textbf{272}$ & \multicolumn{1}{c}{\multirow{3}{*}{$43$}} \\
& Greece & $1.4507$ & Ukraine & $230$ & \\
& Ukraine & $1.4923$ & Greece & $218$ & \\ \midrule
\end{tabular}}
\label{table:eurovision}
\end{table}

\subsubsection{Formula One}

Moreover, we analyze the performance of the voting system employed in the Formula One (F1) world championship. In particular, we imagine that every competing driver constitutes a distinct candidate, while every race in the calendar corresponds to a ``voter''; the ``preferences'' of every race are indicated by the order in which the drivers complete the race. We will assume that when two drivers fail to terminate they will not be comparable (in the spirit of partial orderings). The scoring rule employed in F1 assigns to the first $10$ drivers the points $25, 18, 15, 12, 10, 8, 6, 4, 2, 1$ respectively, and the driver who manages to collect the most number of points throughout the championship is declared the winner; with a slight abuse of notation this rule will also be referred to as the $\textsc{Scoring}$ rule. We will be analyzing the championships held between $2016$ and $2020$, using a \hyperlink{https://www.kaggle.com/aadiltajani/fia-f1-19502019-data}{dataset} from \texttt{Kaggle}. A noteworthy detail is that for the last two years the scoring system assigned an additional point to the driver with the fastest lap, but for simplicity the $\minimax$ will not use any such information. Our results are summarized in \Cref{table:f1}. Again, the driver who won the championship is also the candidate who minimizes distortion, with the exception of $2016$, where---if we are to accept the metric distortion framework---Lewis Hamilton should have won the championship.

\begin{table}[h]
\tiny
\caption{Summary of our findings for the F1 world championships held between $2016$ and $2020$. For every year we have indicated the top three drivers according to the $\minimax$ rule and the $\textsc{Scoring}$ system employed.} 
\centering 
\resizebox{\textwidth}{!}{\begin{tabular}{c | c c | c c | c | c }
\toprule
 \multicolumn{1}{c|}{\multirow{2}{*}{Year}} & \multicolumn{2}{c|}{$\minimax$ rule} & \multicolumn{2}{c|}{\textsc{Scoring} rule} & \multicolumn{1}{c|}{\multirow{2}{*}{$\#$ of Drivers}} & \multicolumn{1}{c}{\multirow{2}{*}{$\#$ of Races}}
\\
& Driver & Distortion & Driver & Score & & \\ \midrule
\multicolumn{1}{c|}{\multirow{3}{*}{$2020$}} & Lewis Hamilton & $\textbf{1.6667}$ & Lewis Hamilton & $\textbf{347}$ & \multicolumn{1}{c|}{\multirow{3}{*}{$23$}} & \multicolumn{1}{c}{\multirow{3}{*}{$17$}} \\
& Valtteri Bottas & $5$ & Valtteri Bottas & $223$ & \\
& Max Verstappen\tablefootnote{Tie with Lando Norris.} & $5.6667$ & Max Verstappen & 214 & \\ \midrule
\multicolumn{1}{c|}{\multirow{3}{*}{$2019$}} & Lewis Hamilton & $\textbf{1.7059}$ & Lewis Hamilton & $\textbf{413}$ & \multicolumn{1}{c|}{\multirow{3}{*}{$20$}} & \multicolumn{1}{c}{\multirow{3}{*}{$21$}} \\
& Valtteri Bottas & $4$ & Valtteri Bottas & $326$ & \\
& Max Verstappen\tablefootnote{Tie with Sebastian Vettel and Charles Leclerc.} & $4.4$ & Max Verstappen & $278$ & \\ \midrule
\multicolumn{1}{c|}{\multirow{3}{*}{$2018$}} & Lewis Hamilton & $\textbf{2}$ & Lewis Hamilton & $\textbf{408}$ & \multicolumn{1}{c|}{\multirow{3}{*}{$20$}} & \multicolumn{1}{c}{\multirow{3}{*}{$21$}} \\
& Sebastian Vettel & $3.6$ & Sebastian Vettel & $320$ & \\
& Kimi Räikkönen & $4.4$ & Kimi Räikkönen & $251$ & \\ \midrule
\multicolumn{1}{c|}{\multirow{3}{*}{$2017$}} & Lewis Hamilton & $\textbf{2.2}$ & Lewis Hamilton & $\textbf{363}$ & \multicolumn{1}{c|}{\multirow{3}{*}{$25$}} & \multicolumn{1}{c}{\multirow{3}{*}{$20$}} \\
& Sebastian Vettel & $3$ & Sebastian Vettel & $317$ & \\
& Valtteri Bottas & $3.1818$ & Valtteri Bottas & $305$ & \\ \midrule
\multicolumn{1}{c|}{\multirow{3}{*}{$2016$}} & Lewis Hamilton & $\textbf{2.8333}$ & Nico Rosberg & $\textbf{385}$ & \multicolumn{1}{c|}{\multirow{3}{*}{$24$}} & \multicolumn{1}{c}{\multirow{3}{*}{$21$}} \\
& Nico Rosberg & $3$ & Lewis Hamilton & $380$ & \\
& Daniel Ricciardo & $3.9091$ & Daniel Ricciardo & $256$ & \\ \midrule
\end{tabular}}
\label{table:f1}
\end{table}

\section{Open Problems}

There are several compelling avenues for future research related to our work. First, it would be interesting to study the performance of the $\DR$ mechanism under randomized pairings; we suspect that this might lead to a substantial improvement since our lower bound (\Cref{proposition:DR-lower_bound}) is very brittle, but we did not pursue this direction. As we previously alluded to, in practice the pairings are typically constructed using some form of prior, so it might be interesting to formalize the guarantees of such techniques. It would also be meaningful to quantify the decay of distortion from $\mathcal{O}(\log m)$ to $\mathcal{O}(1)$ (which is the bound achievable when the mechanism has access to the entire tournament graph) if we gradually elicit more than $m-1$ pairwise comparisons. With regards to the power of deterministic mechanisms which elicit only the $k$-top preferences, an obvious question is to settle \Cref{conjecture:extension}. As we showed in \Cref{proposition:conditional} this would immediately improve our upper bound, but it would still require some further work to close the gap for every value of $k \in [m]$. Finally, can we reduce the sample complexity established in \Cref{theorem:plurality_matching} without sacrificing the efficiency? We argued that the dependence on $\epsilon$ and $\delta$ cannot be improved, but establishing the optimal dependence on the value of $m$ requires future research.

\paragraph{Acknowledgments.} We are very grateful to the anonymous reviewers of SAGT for their insightful comments, and for helping improve the presentation of this work.  

\bibliography{./refs.bib}

\newcommand{\etalchar}[1]{$^{#1}$}
\begin{thebibliography}{AFRSV21}

\bibitem[ABFV20]{DBLP:conf/aaai/AmanatidisBFV20}
Georgios Amanatidis, Georgios Birmpas, Aris Filos{-}Ratsikas, and Alexandros~A.
  Voudouris.
\newblock Peeking behind the ordinal curtain: Improving distortion via cardinal
  queries.
\newblock In {\em The Thirty-Fourth {AAAI} Conference on Artificial
  Intelligence, {AAAI} 2020}, pages 1782--1789. {AAAI} Press, 2020.

\bibitem[ABP15]{DBLP:conf/aaai/AnshelevichBP15}
Elliot Anshelevich, Onkar Bhardwaj, and John Postl.
\newblock Approximating optimal social choice under metric preferences.
\newblock In Blai Bonet and Sven Koenig, editors, {\em Proceedings of the
  Twenty-Ninth {AAAI} Conference on Artificial Intelligence, 2015}, pages
  777--783. {AAAI} Press, 2015.

\bibitem[AFRSV21]{anshelevich2021distortion}
Elliot Anshelevich, Aris Filos-Ratsikas, Nisarg Shah, and Alexandros~A.
  Voudouris.
\newblock Distortion in social choice problems: The first 15 years and beyond,
  2021.

\bibitem[AP16]{DBLP:conf/ijcai/AnshelevichP16}
Elliot Anshelevich and John Postl.
\newblock Randomized social choice functions under metric preferences.
\newblock In Subbarao Kambhampati, editor, {\em Proceedings of the Twenty-Fifth
  International Joint Conference on Artificial Intelligence, {IJCAI} 2016},
  pages 46--59. {IJCAI/AAAI} Press, 2016.

\bibitem[AS16a]{DBLP:conf/aaai/AnshelevichS16}
Elliot Anshelevich and Shreyas Sekar.
\newblock Blind, greedy, and random: Algorithms for matching and clustering
  using only ordinal information.
\newblock In Dale Schuurmans and Michael~P. Wellman, editors, {\em Proceedings
  of the Thirtieth {AAAI} Conference on Artificial Intelligence}, pages
  390--396. {AAAI} Press, 2016.

\bibitem[AS16b]{DBLP:conf/wine/AnshelevichS16}
Elliot Anshelevich and Shreyas Sekar.
\newblock Truthful mechanisms for matching and clustering in an ordinal world.
\newblock In Yang Cai and Adrian Vetta, editors, {\em Web and Internet
  Economics - 12th International Conference, {WINE} 2016}, volume 10123 of {\em
  Lecture Notes in Computer Science}, pages 265--278. Springer, 2016.

\bibitem[AZ18]{DBLP:conf/wine/AnshelevichZ18}
Elliot Anshelevich and Wennan Zhu.
\newblock Ordinal approximation for social choice, matching, and facility
  location problems given candidate positions.
\newblock In George Christodoulou and Tobias Harks, editors, {\em Web and
  Internet Economics - 14th International Conference, {WINE} 2018}, volume
  11316 of {\em Lecture Notes in Computer Science}, pages 3--20. Springer,
  2018.

\bibitem[BCH{\etalchar{+}}15]{DBLP:journals/ai/BoutilierCHLPS15}
Craig Boutilier, Ioannis Caragiannis, Simi Haber, Tyler Lu, Ariel~D. Procaccia,
  and Or~Sheffet.
\newblock Optimal social choice functions: {A} utilitarian view.
\newblock {\em Artif. Intell.}, 227:190--213, 2015.

\bibitem[BCJ11]{Beliakov2011}
Gleb Beliakov, Tomasa Calvo, and Simon James.
\newblock {\em Aggregation of Preferences in Recommender Systems}, pages
  705--734.
\newblock Springer US, 2011.

\bibitem[BDG18]{DBLP:conf/aaai/BhaskarDG18}
Umang Bhaskar, Varsha Dani, and Abheek Ghosh.
\newblock Truthful and near-optimal mechanisms for welfare maximization in
  multi-winner elections.
\newblock In Sheila~A. McIlraith and Kilian~Q. Weinberger, editors, {\em
  Proceedings of the Thirty-Second {AAAI} Conference on Artificial
  Intelligence, (AAAI-18)}, pages 925--932. {AAAI} Press, 2018.

\bibitem[BLP04]{BENFERHAT200425}
Salem Benferhat, Sylvain Lagrue, and Odile Papini.
\newblock Reasoning with partially ordered information in a possibilistic logic
  framework.
\newblock {\em Fuzzy Sets and Systems}, 144(1):25--41, 2004.
\newblock Possibilistic Logic and Related Issues.

\bibitem[BNPS17]{DBLP:conf/aaai/BenadeNP017}
Gerdus Benade, Swaprava Nath, Ariel~D. Procaccia, and Nisarg Shah.
\newblock Preference elicitation for participatory budgeting.
\newblock In Satinder~P. Singh and Shaul Markovitch, editors, {\em Proceedings
  of the Thirty-First {AAAI} Conference on Artificial Intelligence}, pages
  376--382. {AAAI} Press, 2017.

\bibitem[BPQ19]{DBLP:conf/aaai/BenadePQ19}
Gerdus Benad{\`{e}}, Ariel~D. Procaccia, and Mingda Qiao.
\newblock Low-distortion social welfare functions.
\newblock In {\em The Thirty-Third {AAAI} Conference on Artificial
  Intelligence, {AAAI} 2019}, pages 1788--1795. {AAAI} Press, 2019.

\bibitem[Can20]{canonne2020short}
Clément~L. Canonne.
\newblock A short note on learning discrete distributions, 2020.

\bibitem[CE03]{10.2307/30025956}
Sungdai Cho and James~W. Endersby.
\newblock Issues, the spatial theory of voting, and british general elections:
  A comparison of proximity and directional models.
\newblock {\em Public Choice}, 114(3/4):275--293, 2003.

\bibitem[CEG95]{DBLP:journals/ipl/CanettiEG95}
Ran Canetti, Guy Even, and Oded Goldreich.
\newblock Lower bounds for sampling algorithms for estimating the average.
\newblock {\em Inf. Process. Lett.}, 53(1):17--25, 1995.

\bibitem[CFF{\etalchar{+}}16]{DBLP:conf/wine/CaragiannisFFHT16}
Ioannis Caragiannis, Aris Filos{-}Ratsikas, S{\o}ren Kristoffer~Stiil
  Frederiksen, Kristoffer~Arnsfelt Hansen, and Zihan Tan.
\newblock Truthful facility assignment with resource augmentation: An exact
  analysis of serial dictatorship.
\newblock In Yang Cai and Adrian Vetta, editors, {\em Web and Internet
  Economics - 12th International Conference, {WINE} 2016}, volume 10123 of {\em
  Lecture Notes in Computer Science}, pages 236--250. Springer, 2016.

\bibitem[CFNV18]{DBLP:journals/corr/abs-1802-01308}
Ioannis Caragiannis, Aris Filos{-}Ratsikas, Swaprava Nath, and Alexandros~A.
  Voudouris.
\newblock Truthful mechanisms for ownership transfer with expert advice.
\newblock {\em CoRR}, abs/1802.01308, 2018.

\bibitem[CLWA13]{CHEN2013521}
Shuwei Chen, Jun Liu, Hui Wang, and Juan~Carlos Augusto.
\newblock Ordering based decision making – a survey.
\newblock {\em Information Fusion}, 14(4):521--531, 2013.

\bibitem[CNPS16]{DBLP:conf/ijcai/CaragiannisNP016}
Ioannis Caragiannis, Swaprava Nath, Ariel~D. Procaccia, and Nisarg Shah.
\newblock Subset selection via implicit utilitarian voting.
\newblock In Subbarao Kambhampati, editor, {\em Proceedings of the Twenty-Fifth
  International Joint Conference on Artificial Intelligence, {IJCAI} 2016, New
  York, NY, USA, 9-15 July 2016}, pages 151--157. {IJCAI/AAAI} Press, 2016.

\bibitem[DB15]{10.5555/2772879.2773334}
Palash Dey and Arnab Bhattacharyya.
\newblock Sample complexity for winner prediction in elections.
\newblock In {\em Proceedings of the 2015 International Conference on
  Autonomous Agents and Multiagent Systems}, AAMAS '15, page 1421–1430,
  Richland, SC, 2015. International Foundation for Autonomous Agents and
  Multiagent Systems.

\bibitem[FFG16]{DBLP:conf/sigecom/FeldmanFG16}
Michal Feldman, Amos Fiat, and Iddan Golomb.
\newblock On voting and facility location.
\newblock In Vincent Conitzer, Dirk Bergemann, and Yiling Chen, editors, {\em
  Proceedings of the 2016 {ACM} Conference on Economics and Computation, {EC}
  '16}, pages 269--286. {ACM}, 2016.

\bibitem[FFZ14]{DBLP:conf/sagt/Filos-RatsikasF014}
Aris Filos{-}Ratsikas, S{\o}ren Kristoffer~Stiil Frederiksen, and Jie Zhang.
\newblock Social welfare in one-sided matchings: Random priority and beyond.
\newblock In Ron Lavi, editor, {\em Algorithmic Game Theory - 7th International
  Symposium, {SAGT} 2014}, volume 8768 of {\em Lecture Notes in Computer
  Science}, pages 1--12. Springer, 2014.

\bibitem[FGMP19]{DBLP:conf/aaai/FainGMP19}
Brandon Fain, Ashish Goel, Kamesh Munagala, and Nina Prabhu.
\newblock Random dictators with a random referee: Constant sample complexity
  mechanisms for social choice.
\newblock In {\em The Thirty-Third {AAAI} Conference on Artificial
  Intelligence, {AAAI} 2019, The Thirty-First Innovative Applications of
  Artificial Intelligence Conference, {IAAI} 2019}, pages 1893--1900. {AAAI}
  Press, 2019.

\bibitem[FGMS17]{DBLP:conf/wine/FainGMS17}
Brandon Fain, Ashish Goel, Kamesh Munagala, and Sukolsak Sakshuwong.
\newblock Sequential deliberation for social choice.
\newblock In Nikhil~R. Devanur and Pinyan Lu, editors, {\em Web and Internet
  Economics - 13th International Conference, {WINE} 2017}, volume 10660 of {\em
  Lecture Notes in Computer Science}, pages 177--190. Springer, 2017.

\bibitem[FH10]{10.5555/1941934}
Johannes Frnkranz and Eyke Hllermeier.
\newblock {\em Preference Learning}.
\newblock Springer-Verlag, Berlin, Heidelberg, 1st edition, 2010.

\bibitem[FKS21]{DBLP:conf/aistats/FotakisKS21}
Dimitris Fotakis, Alkis Kalavasis, and Konstantinos Stavropoulos.
\newblock Aggregating incomplete and noisy rankings.
\newblock In Arindam Banerjee and Kenji Fukumizu, editors, {\em The 24th
  International Conference on Artificial Intelligence and Statistics, {AISTATS}
  2021, April 13-15, 2021, Virtual Event}, volume 130 of {\em Proceedings of
  Machine Learning Research}, pages 2278--2286. {PMLR}, 2021.

\bibitem[FV20]{DBLP:journals/corr/abs-2007-06304}
Aris Filos{-}Ratsikas and Alexandros~A. Voudouris.
\newblock Approximate mechanism design for distributed facility location.
\newblock {\em CoRR}, abs/2007.06304, 2020.

\bibitem[GAX17]{DBLP:conf/aaai/GrossAX17}
Stephen Gross, Elliot Anshelevich, and Lirong Xia.
\newblock Vote until two of you agree: Mechanisms with small distortion and
  sample complexity.
\newblock In Satinder~P. Singh and Shaul Markovitch, editors, {\em Proceedings
  of the Thirty-First {AAAI} Conference on Artificial Intelligence, 2017},
  pages 544--550. {AAAI} Press, 2017.

\bibitem[GHS20]{DBLP:conf/focs/Gkatzelis0020}
Vasilis Gkatzelis, Daniel Halpern, and Nisarg Shah.
\newblock Resolving the optimal metric distortion conjecture.
\newblock In {\em 61st {IEEE} Annual Symposium on Foundations of Computer
  Science, {FOCS} 2020}, pages 1427--1438. {IEEE}, 2020.

\bibitem[GKM17]{10.1145/3033274.3085138}
Ashish Goel, Anilesh~K. Krishnaswamy, and Kamesh Munagala.
\newblock Metric distortion of social choice rules: Lower bounds and fairness
  properties.
\newblock In {\em Proceedings of the 2017 ACM Conference on Economics and
  Computation}, EC '17, page 287–304. Association for Computing Machinery,
  2017.

\bibitem[GO21]{gurobi}
LLC Gurobi~Optimization.
\newblock Gurobi optimizer reference manual, 2021.

\bibitem[Kar84]{10.1145/800057.808695}
Narendra Karmarkar.
\newblock A new polynomial-time algorithm for linear programming.
\newblock {\em Comb.}, 4(4):373--396, 1984.

\bibitem[Kem20a]{DBLP:conf/aaai/000120a}
David Kempe.
\newblock An analysis framework for metric voting based on {LP} duality.
\newblock In {\em The Thirty-Fourth {AAAI} Conference on Artificial
  Intelligence, {AAAI} 2020}, pages 2079--2086. {AAAI} Press, 2020.

\bibitem[Kem20b]{DBLP:conf/aaai/000120b}
David Kempe.
\newblock Communication, distortion, and randomness in metric voting.
\newblock In {\em The Thirty-Fourth {AAAI} Conference on Artificial
  Intelligence, {AAAI} 2020}, pages 2087--2094. {AAAI} Press, 2020.

\bibitem[KR08]{10.1214/07-AOP384}
Leonid~(Aryeh) Kontorovich and Kavita Ramanan.
\newblock {Concentration inequalities for dependent random variables via the
  martingale method}.
\newblock {\em The Annals of Probability}, 36(6):2126 -- 2158, 2008.

\bibitem[LN95]{10.1257/jep.9.1.3}
Jonathan Levin and Barry Nalebuff.
\newblock An introduction to vote-counting schemes.
\newblock {\em Journal of Economic Perspectives}, 9(1):3--26, March 1995.

\bibitem[LN05]{10.1007/11510888_21}
Ga{\"e}lle Legrand and Nicolas Nicoloyannis.
\newblock Feature selection method using preferences aggregation.
\newblock In Petra Perner and Atsushi Imiya, editors, {\em Machine Learning and
  Data Mining in Pattern Recognition}, pages 203--217, Berlin, Heidelberg,
  2005. Springer Berlin Heidelberg.

\bibitem[LPR{\etalchar{+}}07]{DBLP:conf/ijcai/LangPRVW07}
J{\'{e}}r{\^{o}}me Lang, Maria~Silvia Pini, Francesca Rossi, Kristen~Brent
  Venable, and Toby Walsh.
\newblock Winner determination in sequential majority voting.
\newblock In Manuela~M. Veloso, editor, {\em {IJCAI} 2007, Proceedings of the
  20th International Joint Conference on Artificial Intelligence}, pages
  1372--1377, 2007.

\bibitem[MAL57]{10.1093/biomet/44.1-2.114}
C.~L. MALLOWS.
\newblock {NON-NULL RANKING MODELS. I}.
\newblock {\em Biometrika}, 44(1-2):114--130, 06 1957.

\bibitem[Mou86]{10.2307/41105842}
H.~Moulin.
\newblock Choosing from a tournament.
\newblock {\em Social Choice and Welfare}, 3(4):271--291, 1986.

\bibitem[MPSW19]{DBLP:conf/nips/MandalP0W19}
Debmalya Mandal, Ariel~D. Procaccia, Nisarg Shah, and David~P. Woodruff.
\newblock Efficient and thrifty voting by any means necessary.
\newblock In Hanna~M. Wallach, Hugo Larochelle, Alina Beygelzimer, Florence
  d'Alch{\'{e}}{-}Buc, Emily~B. Fox, and Roman Garnett, editors, {\em Advances
  in Neural Information Processing Systems 32: Annual Conference on Neural
  Information Processing Systems 2019}, pages 7178--7189, 2019.

\bibitem[MSW20]{10.1145/3391403.3399510}
Debmalya Mandal, Nisarg Shah, and David~P. Woodruff.
\newblock Optimal communication-distortion tradeoff in voting.
\newblock In {\em Proceedings of the 21st ACM Conference on Economics and
  Computation}, EC '20, page 795–813. Association for Computing Machinery,
  2020.

\bibitem[MW19]{DBLP:conf/ec/MunagalaW19}
Kamesh Munagala and Kangning Wang.
\newblock Improved metric distortion for deterministic social choice rules.
\newblock In Anna Karlin, Nicole Immorlica, and Ramesh Johari, editors, {\em
  Proceedings of the 2019 {ACM} Conference on Economics and Computation, {EC}
  2019}, pages 245--262. {ACM}, 2019.

\bibitem[PR06]{DBLP:conf/cia/ProcacciaR06}
Ariel~D. Procaccia and Jeffrey~S. Rosenschein.
\newblock The distortion of cardinal preferences in voting.
\newblock In Matthias Klusch, Michael Rovatsos, and Terry~R. Payne, editors,
  {\em Cooperative Information Agents X, 10th International Workshop, {CIA}
  2006}, volume 4149 of {\em Lecture Notes in Computer Science}, pages
  317--331. Springer, 2006.

\bibitem[PZPR09]{DBLP:journals/ai/ProcacciaZPR09}
Ariel~D. Procaccia, Aviv Zohar, Yoni Peleg, and Jeffrey~S. Rosenschein.
\newblock The learnability of voting rules.
\newblock {\em Artif. Intell.}, 173(12-13):1133--1149, 2009.

\bibitem[SE17]{DBLP:conf/aaai/SkowronE17}
Piotr~Krzysztof Skowron and Edith Elkind.
\newblock Social choice under metric preferences: Scoring rules and {STV}.
\newblock In Satinder~P. Singh and Shaul Markovitch, editors, {\em Proceedings
  of the Thirty-First {AAAI} Conference on Artificial Intelligence, February
  4-9, 2017, San Francisco, California, {USA}}, pages 706--712. {AAAI} Press,
  2017.

\bibitem[SM96]{10.2307/25054952}
Donald~G. Saari and Vincent~R. Merlin.
\newblock The copeland method: I.: Relationships and the dictionary.
\newblock {\em Economic Theory}, 8(1):51--76, 1996.

\bibitem[TRG03]{10.2307/41106568}
Ilia Tsetlin, Michel Regenwetter, and Bernard Grofman.
\newblock The impartial culture maximizes the probability of majority cycles.
\newblock {\em Social Choice and Welfare}, 21(3):387--398, 2003.

\bibitem[VLZ12]{DBLP:conf/cikm/VolkovsLZ12}
Maksims Volkovs, Hugo Larochelle, and Richard~S. Zemel.
\newblock Learning to rank by aggregating expert preferences.
\newblock In Xue{-}wen Chen, Guy Lebanon, Haixun Wang, and Mohammed~J. Zaki,
  editors, {\em 21st {ACM} International Conference on Information and
  Knowledge Management, CIKM'12, 2012}, pages 843--851. {ACM}, 2012.

\bibitem[vNM44]{vonNeumann1944-VONTOG-4}
John von Neumann and Oskar Morgenstern.
\newblock {\em Theory of Games and Economic Behavior}.
\newblock Princeton University Press, 1944.

\bibitem[VZ14]{DBLP:journals/jmlr/VolkovsZ14}
Maksims Volkovs and Richard~S. Zemel.
\newblock New learning methods for supervised and unsupervised preference
  aggregation.
\newblock {\em J. Mach. Learn. Res.}, 15(1):1135--1176, 2014.

\end{thebibliography}

\appendix

\section{Lower Bound for \texorpdfstring{$k$}{}-top Preferences}
\label{appendix:lower_bound}

In this section we present for completeness a lower bound for deterministic mechanisms which have access only to the $k$-top preferences of every voter. 

\begin{proposition}
    \label{proposition:k_top-lower_bound}
    Any deterministic mechanism which elicits only the $k$-top preferences from every voter out of the $m$ alternatives has distortion $\Omega(m/k)$.
\end{proposition}

\begin{proof}
First of all, assume without any loss of generality that $k \mid (m-1)$\footnote{If it is not the case that $k \mid (m-1)$ take $k'$ to be the smallest number larger than $k$ such that $k' \mid (m-1)$, and apply our argument for $k'$; given that $k' < 2k$ we will establish again a lower bound of $\Omega(m/k)$ even though the mechanism had more information than the $k$-top preferences.}, and let $n = (m-1)/k$ be the number of voters. For simplicity, let us enumerate the number of candidates as $C = \{1, 2, \dots, n \times k\} \cup \{ x\}$. Now consider some preference profile $\mathcal{P}$ in which the $k$-top preferences of voter $i$ correspond to the of candidates $\{(i-1)k + 1, \dots, (i-1) k + k\}$ according to some arbitrary order; observe that all of these sets are pairwise disjoint. 

Based on these preferences the mechanism has to select some candidate. If $x$ is selected the lower bound follows trivially since $x$ could actually be the last choice for every voter. Therefore, let us assume that candidate $1$ was selected by the mechanism; this hypothesis is without loss of generality due to the symmetry of the input $\mathcal{P}$. However, the agents and the candidates could be located on the metric space of \Cref{fig:lower_bound}; indeed, it is easy to check that the induced metric space is consistent with the given preferences. As a result, it follows that 

\begin{equation}
    \frac{\SC(1)}{\SC(x)} = \frac{D + (n-1) \times (\delta + 2D) }{D + (n-1) \times \delta} = \frac{1 + (n-1) \times (\delta/D + 2) }{1 + (n-1) \times \delta/D}.
\end{equation}
Thus, for $\delta/D \downarrow 0$ we obtain that $\SC(1)/\SC(x) \to 2n - 1 = \Omega(m/k)$.
\end{proof}

\begin{figure}[!ht]
    \centering
    \includegraphics[scale=0.35]{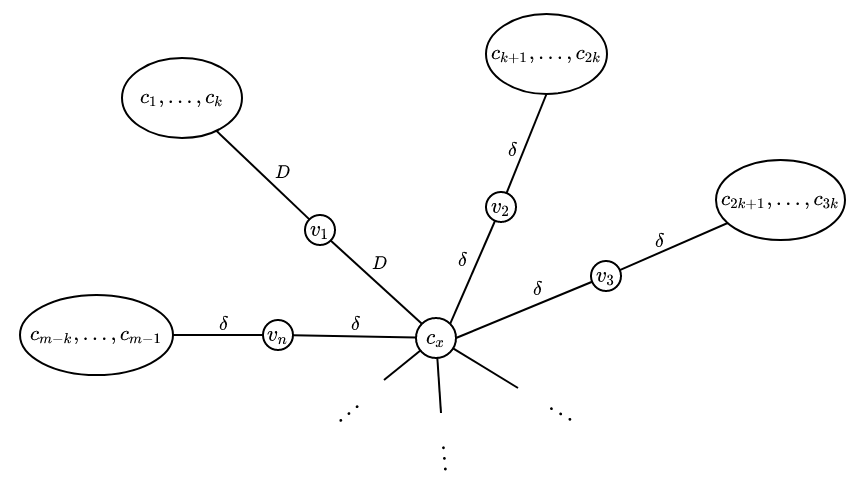}
    \caption{The metric space considered for the proof of \Cref{proposition:k_top-lower_bound}, where $\delta/D \downarrow 0$ for some positive numbers $\delta$ and $D$. Naturally, the distance between two points is simply the shortest path in the graph.}
    \label{fig:lower_bound}
\end{figure}

We should note that although in our worst-case example the number of voters $n$ is smaller than the number of candidates $m$, which is not the canonical case, our argument directly extends whenever $n$ is a multiple of $(m-1)/k$, allowing $n$ to be arbitrarily large. Moreover, a similar construction shows an $\Omega(m/k)$ lower bound for $\alpha$-\emph{decisive} metrics \cite{DBLP:conf/ijcai/AnshelevichP16}, for any $\alpha \in [0,1]$; indeed, it suffices to place the voters within the ``cluster'' of their $k$-most preferred candidates.

\section{\texorpdfstring{$\pluralitymatching$}{} vs \texorpdfstring{$\minimax$}{}}

In this section we compare the $\pluralitymatching$ mechanism of Gkatzelis et al. \cite{DBLP:conf/focs/Gkatzelis0020} with the instance-optimal mechanism, namely $\minimax$; in this section we tacitly posit that $\mathcal{P} = \sigma$, i.e. all the agents provide their entire rankings to the mechanism. 

\subsection{Instance Optimality}

The first question that arises is how far could the distortion of $\pluralitymatching$ be with respect to the instance-optimal candidate; to this end, we commence with the following proposition:

\begin{lemma}[Lemma 6, \cite{DBLP:conf/focs/Gkatzelis0020}]
    \label{lemma:veto-plu}
For any election $\mathcal{E} = (V, C, \sigma)$, a candidate $a \in C$ can be selected by $\pluralitymatching$ only if $\plu(a) \geq \veto(a)$.
\end{lemma}

With this lemma in mind, we consider an instance with a set of $m$ voters $V = \{1, 2, \dots, m\}$, and a set of $m$ candidates $C = \{a, \dots\}$. We assume that for every voter $i \in [n-1], \secc(i) = a$, while the (single) top-preferences of all the voters $i \in [n-1]$ are assumed to be pairwise-distinct. Finally, the last voter places candidate $a$ in her last place, while her preferences are otherwise arbitrary. An example with four candidates $\{a, b, e, f\}$ corresponds to the following input:

\begin{itemize}
    \item $b \succ_1 a \succ_1 e \succ_1 f$;
    \item $f \succ_2 a \succ_2 b \succ_2 e$;
    \item $e \succ_3 a \succ_3 f \succ_3 b$;
    \item $b \succ_4 f \succ_4 e \succ_4 a$.
\end{itemize}

In general, observe that for any candidate $b \in C \setminus \{ a \}$ it follows that $|ab| = (m-2)/m$. Moreover, we will use the following standard lemma:

\begin{lemma}
Consider two (distinct) candidates $a, b \in C$ such that $|ab| \geq \alpha n > 0$. Then, 

\begin{equation}
    \frac{\SC(a)}{\SC(b)} \leq \frac{2}{\alpha} - 1.
\end{equation}
\end{lemma}

This implies that the distortion of candidate $a$ is $1 + \mathcal{O}(1/m)$. However, given that $\plu(a) = 0 < 1 = \veto(a)$, we know from \Cref{lemma:veto-plu} that $a$ cannot be selected by $\pluralitymatching$. We will show that every other candidate yields distortion close to $3$. In particular, consider the metric space illustrated in \Cref{fig:lower_bound-instance_opt}. It is easy to verify that the induced metric space is consistent with the given preferences. But, it follows that $\SC(a) = m$, while $\SC(b) = 2 + 3(m-2) = 3m - 4$ for any $b \neq a$, implying that $\SC(b)/\SC(a) = 3 - \mathcal{O}(1/m)$. As a result, we have arrived at the following conclusion:

\begin{proposition}
    \label{proposition:PL-vs-minimax}
For any sufficiently small $\epsilon > 0$ and $m = \mathcal{O}(1/\epsilon)$ there exists a preferences profile $\sigma$ such that $\minimax$ yields distortion $1 + \epsilon$, while $\pluralitymatching$ yields distortion at least $3 - \epsilon$.
\end{proposition}

\begin{figure}[!ht]
    \centering
    \includegraphics[scale=0.35]{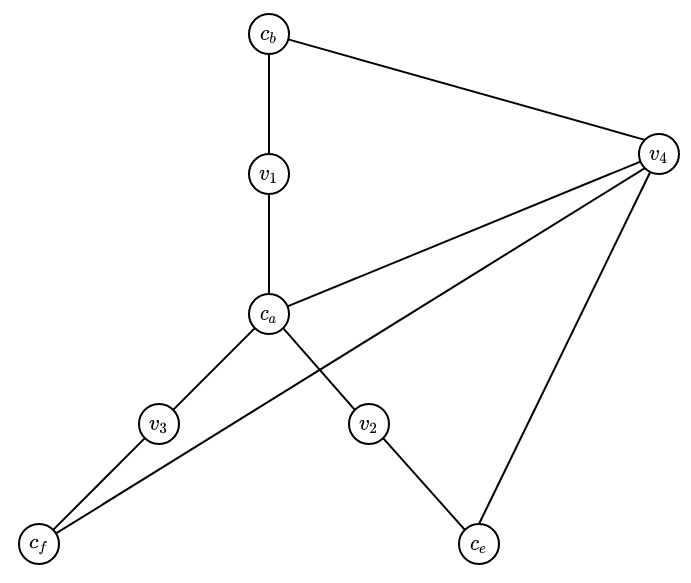}
    \caption{A metric space embedded on an \emph{unweighted} and \emph{undirected} graph; this example corresponds to $m=n=4$, but the pattern should be already clear.}
    \label{fig:lower_bound-instance_opt}
\end{figure}

\subsection{Decisive Metrics}

Moreover, it is natural to compare these mechanisms in more refined metrics. Specifically, we espouse the $\alpha$-decisiveness assumption of Anshelevich and Postl \cite{DBLP:conf/ijcai/AnshelevichP16}, according to which $d(v_i, c_p) \leq \alpha \cdot d(v_i, c_q)$, where $p = \topp(i)$ and $q = \secc(i)$, and $\alpha \in [0,1]$ some parameter; notice that the general case corresponds to $\alpha = 1$, while for $\alpha = 0$ every voter also serves as a candidate. The first observation is that this particular refinement can be addresses by simply incorporating some additional constraints in the linear program. More precisely, for a pair of distinct candidates $a, b$ this leads to the following linear program $\metricalphaLP(a, b)$:

\begin{equation}
    \label{eq:linear_program-alpha}
\begin{array}{ll@{}ll}
\text{maximize}  & \sum_{i=1}^n x_{i, a} &\\
\text{subject to} & \sum_{i=1}^n x_{i, b} = 1; \\ 
                & x_{i, \topp(i)} \leq \alpha \cdot x_{i, \secc(i)}, & \forall i \in V; \\
                & x_{i, p} \leq x_{i, q}, & \forall (p, q) \in \mathcal{P}_i, \forall i \in V; \\
                 &x_{i, i} = 0, & \forall i \in V \cup C; \\
                 & x_{i, j} = x_{j, i}, & \forall i, j \in V \cup C; \\
                 & x_{i, j} \leq x_{i, k} + x_{k, j}, & \forall i, j, k \in V \cup C.
\end{array}
\end{equation}

Here we have assumed that every agent $i$ provides her most preferred candidate $\topp(i)$, as well as her second most preferred candidate $\secc(i)$. Having solved the $\metricalphaLP(a,b)$ for every distinct pair of candidates $a,b$, we simply select the candidate who minimizes the maximum cost obtained over all other candidates; this mechanism shall be referred to as the $\minimaxalpha$. Similarly to \Cref{theorem:instance_opt} we can establish the following:

\begin{proposition}
    \label{proposition:instance_opt-decisive}
For any given preferences $\mathcal{P}$ and any $\alpha \in [0,1]$ the $\minimaxalpha$ rule is instance-optimal in terms of distortion under $\alpha$-decisive metrics.
\end{proposition}

We should point out that for $\alpha$-decisive metrics $\pluralitymatching$ always yields a candidate with distortion $2 + \alpha$. Moreover, Gkatzelis et al. \cite{DBLP:conf/focs/Gkatzelis0020} showed a lower bound of $2 + \alpha - 2(1- \alpha)/m'$ for deterministic mechanisms, where $m' = 2 \lfloor m/2 \rfloor$; thus, they showed that their mechanism obtains the optimal distortion only when $m \to \infty$ or when $\alpha = 1$, leaving a substantial gap.

We will show that $\minimaxalpha$ can substantially outperform $\pluralitymatching$ even for $\alpha$-decisive metrics with $\alpha$ close to $0$. Specifically, consider an election with $3$ candidates and $2$ voters\footnote{This example is taken from \cite{DBLP:conf/focs/Gkatzelis0020}.} with the following preferences: $\sigma_1 = a \succ b \succ e$, and $\sigma_2 = e \succ b \succ a$. For this election, $b$ could be returned by $\pluralitymatching$ (see \cite{DBLP:conf/focs/Gkatzelis0020}). However, we claim that $b$ yields distortion $2 + \alpha$, while $a$ and $e$ have distortion $1 + 2\alpha$. Indeed, we will show that candidate $a$ has always distortion upper-bounded by $1 + 2\alpha$ (by symmetry, the same holds for $e$), while for candidate $b$ there exists a metric space for which $b$ yields distortion $2 + \alpha$. Specifically, we have that $d(c_a, c_b) \leq d(v_1, c_a) + d(v_1, c_b) \leq (1 + \alpha) d(v_1, c_b)$; thus we obtain that

\begin{equation}
d(v_1,c_a) \leq \alpha d(v_1, c_b),    
\end{equation}

\begin{equation}
    d(v_2, c_a) \leq d(v_2, c_b) + d(c_a, c_b) \leq (1 + \alpha) d(v_1, c_b) + d(v_2, c_b).
\end{equation}

Summing these inequalities yields that $\SC(a) \leq (1+2\alpha) d(v_1, c_b) + d(v_2, c_b) \leq (1+2\alpha) \SC(b)$. Similarly, we can prove that $\SC(a) \leq (1+2\alpha) \SC(e)$. On the other hand, for candidate $b$ Gkatzelis et al. \cite{DBLP:conf/focs/Gkatzelis0020} considered the following metric space:

\begin{figure}[!ht]
    \centering
    \includegraphics[scale=0.5]{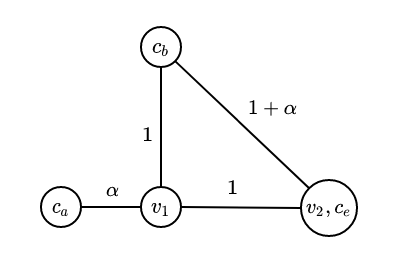}
    \caption{A metric space embedded on a graph.}
    \label{fig:decisive}
\end{figure}

Naturally, the distance between a pair of nodes is the corresponding shortest path in the graph. Thus, for this instance it follows that $\SC(e) = 1$, while $\SC(b) = 2 + \alpha$, implying that the distortion of $b$ is $2 + \alpha$. Thus, for $\alpha \to 0$ $\pluralitymatching$ loses a factor of $2$ with respect to the optimal candidate, which would be identified by the $\minimaxalpha$ rule by virtue of \Cref{proposition:instance_opt-decisive}.

\begin{proposition}
There exists a preference profile $\sigma$ such that $\minimaxalpha$ yields distortion $1 + 2\alpha$, while $\pluralitymatching$ yields distortion at least $2 + \alpha$ under $\alpha$-decisive metrics.
\end{proposition}

Nonetheless we should point out that $\pluralitymatching$ does \emph{not} require knowing the value of parameter $\alpha$, unlike the instance-optimal mechanism.

\section{Additional Experiments}
\label{appendix:additional_experiments}

In this section we provide additional details about our empirical findings. Specifically, in \Cref{fig:eurovision} we juxtapose the scores of the finalists in the Eurovision song contest (based on the $\textsc{Scoring}$ rule) with their distortion as determined by $\minimax$. We note that we have removed the following outliers:

\begin{itemize}
    \item For the year $2006$ the countries Malta and Spain which incurred a distortion of $24.3333$ and $18.0000$ respectively.
    \item For the year $2007$ the countries Ireland and United Kingdom which incurred a distortion of $27.0000$ and $13.6667$ respectively.
    \item For the year $2008$ the countries United Kingdom, Germany and Poland which incurred a distortion of $20.5000$, $20.5000$ and $14$ respectively.
\end{itemize}

Similarly, in \Cref{fig:f1} we juxtapose the scores accumulated by the drives based on the $\textsc{Scoring}$ rule with their distortion as determined by $\minimax$. These results follow after removing the following outliers:

\begin{itemize}
    \item For the year $2016$ the driver Stoffel Vandoorne who incurred a distortion of $41$. 
    \item For the year $2017$ the drivers Paul di Resta, Jenson Button, Brendon Hartley, and Antonio Giovinazzi who incurred a distortion of $\infty$, $\infty$, $19$, and $39$ respectively.
    \item For the year $2020$ the drivers Pietro Fittipaldi, Jack Aitken, and Nico Hulkenberg who incurred a distortion of $16$, $33$, and $16$ respectively.
\end{itemize}

\begin{figure}
    \centering
    \includegraphics[scale=0.45]{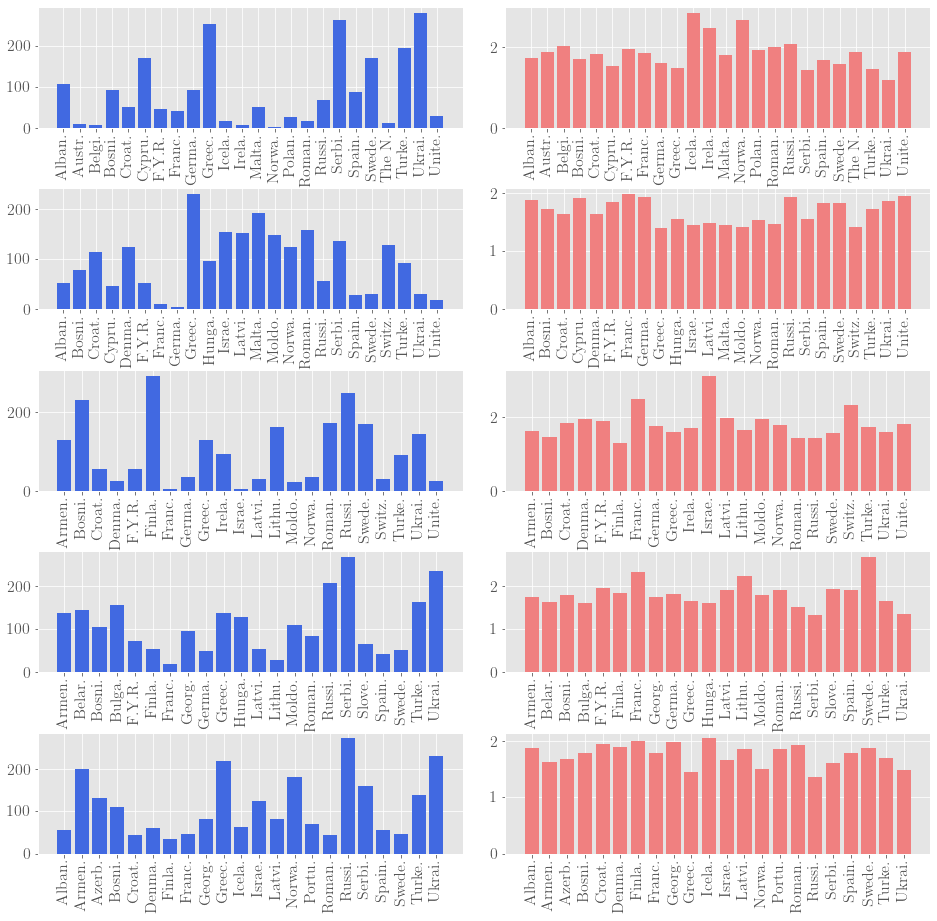}
    \caption{$\textsc{Scoring}$ rule vs $\minimax$ for the Eurovision song contests during the years between $2004$ and $2008$. For every country we have indicated only the first $5$ letters according to the entry in the dataset.}
    \label{fig:eurovision}
\end{figure}

\begin{figure}
    \centering
    \includegraphics[scale=0.45]{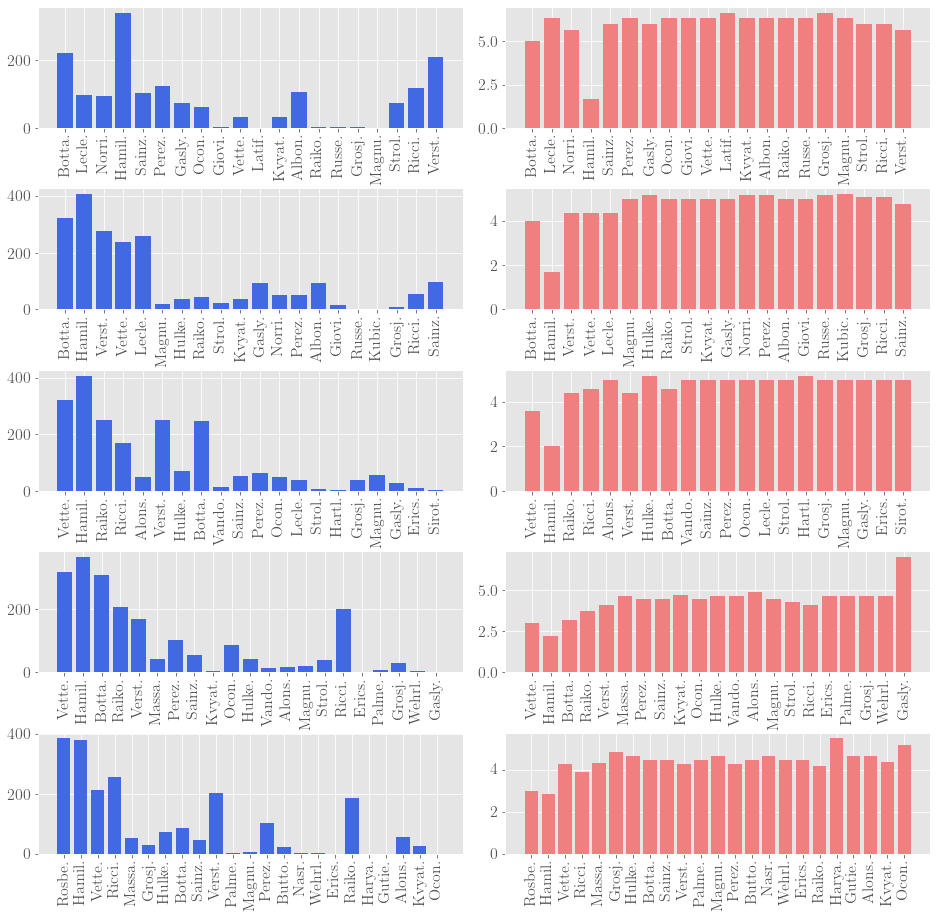}
    \caption{$\textsc{Scoring}$ rule vs $\minimax$ for the F1 world championships during the years between $2020$ and $2016$. For every driver we have indicated only the first $5$ letters of his/her (last) name according to the entry in the dataset.}
    \label{fig:f1}
\end{figure}

\end{document}